\documentclass[11pt,letterpaper]{article}
\usepackage{color,ifpdf,latexsym,graphicx,url,amsthm,amsmath,amsfonts,subcaption,float}

\usepackage{wasysym}
\usepackage{placeins}
\usepackage{thmtools,thm-restate}

\usepackage{hyperref}

\advance \topmargin by -\headheight
\advance \topmargin by -\headsep
\evensidemargin \oddsidemargin
\newcommand{\defn}[1]{\textbf{\textit{#1}}}
\newcommand{\prob}[1]{\textsc{#1}}

{\makeatletter \gdef\fps@figure{!htbp}}

\let\realbfseries=\bfseries
\def\bfseries{\realbfseries\boldmath}

\newtheorem{theorem}{Theorem}[section]
\newtheorem{lemma}[theorem]{Lemma}
\newtheorem{corollary}[theorem]{Corollary}
\newtheorem{definition}[theorem]{Definition}
\newtheorem{problem}{Problem}

{\makeatletter
 \gdef\xxxmark{%
   \expandafter\ifx\csname @mpargs\endcsname\relax 
     \expandafter\ifx\csname @captype\endcsname\relax 
       \marginpar{xxx}
     \else
       xxx 
     \fi
   \else
     xxx 
   \fi}
 \gdef\xxx{\@ifnextchar[\xxx@lab\xxx@nolab}
 \long\gdef\xxx@lab[#1]#2{\textbf{[\xxxmark #2 ---{\sc #1}]}}
 \long\gdef\xxx@nolab#1{\textbf{[\xxxmark #1]}}
}

\makeatletter
\def\GrabProofArgument[#1]{ #1: \egroup\ignorespaces}
\def\proof{\noindent\textbf\bgroup Proof%
           \@ifnextchar[{\GrabProofArgument}{: \egroup\ignorespaces}}

\makeatother

\begin{document}

\title{Hamiltonicity is Hard in Thin or Polygonal Grid Graphs, but Easy in Thin Polygonal Grid Graphs}
\author{
Erik D. Demaine%
    \thanks{MIT Computer Science and Artificial Intelligence Laboratory,
      32 Vassar St., Cambridge, MA 02139, USA,
      \protect\url{edemaine@mit.edu}}
\and
  Mikhail Rudoy%
  \thanks{Google Inc, \protect\url{mrudoy@gmail.com}. Work completed at MIT Computer Science and Artificial Intelligence Laboratory.}
}
\date{}

\maketitle

\begin{abstract}
In 2007, Arkin et al.~\cite{Arkin} initiated a systematic study of the complexity of the Hamiltonian cycle problem on square, triangular, or hexagonal grid graphs, restricted to polygonal, thin, superthin, degree-bounded, or solid grid graphs. They solved many combinations of these problems, proving them either polynomially solvable or NP-complete, but left three combinations open.
In this paper, we prove two of these unsolved combinations to be NP-complete:
\prob{Hamiltonicity of Square Polygonal Grid Graphs} and
\prob{Hamiltonicity of Hexagonal Thin Grid Graphs}.
We also consider a new restriction, where the grid graph is both thin and
polygonal, and prove that Hamiltonicity then becomes polynomially solvable for
square, triangular, and hexagonal grid graphs.
%
%
%
\end{abstract}

\FloatBarrier
\section{Introduction}
\FloatBarrier

Hamiltonicity (Hamiltonian Cycle) is one of the prototype NP-complete problems from Karp's 1972
paper \cite{Karp-1972}.  An important NP-complete special case of Hamiltonicity
is its restriction to \emph{(square) grid graphs}
\cite{Itai-Papadimitriou-Luiz-1982},
where vertices lie on the 2D integer square grid and edges connect all
unit-distance vertex pairs.
Hamiltonicity in grid graphs has been the basis for NP-hardness reductions to
many geometric and planar-graph problems, such as
Euclidean TSP \cite{Itai-Papadimitriou-Luiz-1982},
Euclidean degree-bounded minimum spanning tree
\cite{Papadimitriou-Vazirani-1984},
2D platform games with item collection and time limits \cite{Forisek-2010},
the Slither Link puzzle \cite{Yato-2000},
the Hashiwokakero puzzle \cite{Andersson-2009},
lawn mowing and milling (e.g., 3D printing) \cite{Arkin-Fekete-Mitchell-2000},
and
minimum-turn milling \cite{Arkin-Bender-Demaine-Fekete-Mitchell-Sethia-2005};
see \cite{6.890}.

Given all these applications, it is natural to wonder how special we can make
the grid graphs, and whether we can change the grid to triangular or hexagonal,
and still keep Hamiltonicity NP-complete.  Two notable examples are
NP-completeness in maximum-degree-3 square grid graphs
\cite{Papadimitriou-Vazirani-1984} and a polynomial-time algorithm for
\emph{solid} square grid graphs \cite{Umans-Lenhart-1997}.
In 2007, Arkin et al.~\cite{Arkin} initiated a systematic study of the
complexity of Hamiltonicity in square, triangular, or hexagonal grid graphs,
restricted to several special cases: polygonal, thin, superthin,
degree-bounded, or solid grid graphs.
See \cite{Arkin} or Section~\ref{sec:terminology}
for definitions.
Table~\ref{tab:results_table} (nonbold) summarizes the many results they
obtained, including several NP-completeness results and a few polynomial-time
algorithms.

\paragraph*{Our Results}
Arkin et al.~\cite{Arkin} left unsolved three of the combinations between grid
shape and special property: \prob{Hamiltonicity of Square Polygonal Grid
Graphs}, \prob{Hamiltonicity of Hexagonal Thin Grid Graphs}, and
\prob{Hamiltonicity of Hexagonal Solid Grid Graphs}.  In this paper, we
prove that the first two of these, \prob{Hamiltonicity of Square Polygonal Grid
Graphs} and \prob{Hamiltonicity of Hexagonal Thin Grid Graphs}, are NP-complete.
In addition, we consider another case not considered in that paper, namely,
\emph{thin polygonal} grid graphs (the fusion of two special cases).
We show that Hamiltonicity becomes polynomially solvable in this case,
for all three shapes of grid graph.

Table~\ref{tab:results_table} (bold) summarizes our new results.

\begin{table}
\centering
\begin{tabular}{ | l | c | c | c |}
\hline
Grid & Triangular & Square & Hexagonal \\
\hline \hline
General & NP-complete & NP-complete & NP-complete \\
\hline
Degree- & deg $\le 3$ & deg $\le 3$ & deg $\le 2$ \\
bounded & NP-complete & NP-complete & Polynomial \\
\hline
Thin & NP-complete & NP-complete & \textbf{NP-complete} \\
\hline
Superthin & NP-complete & Polynomial & Polynomial \\
\hline
Polygonal & Polynomial & \textbf{NP-complete} & NP-complete \\
\hline
Solid & Polynomial & Polynomial & Open \\
\hline
Thin Polygonal & \textbf{Polynomial} & \textbf{Polynomial} & \textbf{Polynomial} \\
\hline
\end{tabular}
\caption{Complexity of Hamiltonicity in grid graph variants; bold entries correspond to new results in this paper (see \cite{Arkin} or Section~\ref{sec:terminology} for definitions)}
\label{tab:results_table}
\end{table}

In Section~\ref{sec:terminology}, we briefly define the several types of grid graphs. In Section~\ref{sec:polygonal thin}, we show that Hamiltonicity can be solved in polynomial time in the three thin polygonal grid graph cases; this is particularly challenging for hexagonal grid graphs, where the problem reduces to the polynomially solvable problem \prob{Max-Degree-$3$ Tree-Residue Vertex-Breaking}. In Section~\ref{sec:hex thin}, we prove NP-completeness of \prob{Hamiltonicity of Hexagonal Thin Grid Graphs}. In Section~\ref{sec:square polygonal}, we prove NP-completeness of the \prob{Hamiltonicity of Square Polygonal Grid Graphs} problem. Finally, in Section~\ref{sec:conclusion}, we discuss the final remaining open problem, \prob{Hamiltonicity of Hexagonal Solid Grid Graphs}.

\FloatBarrier
\section{Grid Graph Terminology}
\FloatBarrier
\label{sec:terminology}

In this section, we introduce the definitions of several terms relating to grid graphs. We restrict our attention to only those terms and concepts relevant to the contents of this paper. See Arkin et al.~\cite{Arkin} for a more general overview of these concepts. 

We begin with a general definition.

\begin{definition}
The sets $\mathbb{Z}_{\square}$, $\mathbb{Z}_{\triangle}$, and $\mathbb{Z}_{\hexagon}$ refer to the sets of vertices of the tilings of the plane with unit-side squares, equilateral triangles, and regular hexagons. A \defn{grid graph} is a finite graph $G = (V, E)$ where $V$ is a subset of $\mathbb{Z}_{\square}$, $\mathbb{Z}_{\triangle}$, or $\mathbb{Z}_{\hexagon}$ and $E$ is the set of pairs $(u,v)$ of elements of $V$ such that $u$ and $v$ are at a distance of $1$ from each other. If $V \subset \mathbb{Z}_{\square}$, the grid graph is said to be a \defn{square grid graph}. Similarly, if $V \subset \mathbb{Z}_{\triangle}$ then $G$ is said to be a \defn{triangular grid graph} and if $V \subset \mathbb{Z}_{\hexagon}$ then $G$ is said to be a \defn{hexagonal grid graph}.
\end{definition}

Because we are concerned with Hamiltonicity, we restrict our attention to connected grid graphs with no degree-1 vertices. This does not affect the hardness of any Hamiltonicity problems because all grid graphs which are disconnected or which contain a degree-1 vertex are trivially not Hamiltonian and can be easily recognized.

In order to define the grid graph properties we are interested in, we need some more terminology:

\begin{definition}
Let $G$ be a grid graph. Consider the faces of the graph. There is one unbounded face. The cycle bordering this unbounded face is called the \defn{outer boundary} of $G$. The bounded faces of $G$ fall into two categories. Any bounded face containing a lattice point in its interior is called a \defn{hole}. The cycles bordering the holes of $G$ are called the \defn{inner boundaries} of $G$. The other category of bounded face is the category without lattice points in the interior; any such face must necessarily have a minimal length cycle (length 3, 4, or 6 for triangular, square, or hexagonal grid graphs) as its boundary. This type of face is called a \defn{pixel}. Any vertex on the inner or outer boundaries is called a \defn{boundary vertex}. All other vertices are \defn{interior vertices}.
\end{definition}

The above terminology allows us to define the grid graph properties of interest:

\begin{definition}
A \defn{polygonal grid graph} is a grid graph $G = (V, E)$ such that every vertex in $V$ and every edge in $E$ belongs to a pixel and such that no vertex can be removed to merge two boundaries (See Figure~\ref{fig:graphs}, top.)
\end{definition}

\begin{figure}[!htbp]
  \centering
  \begin{tabular}{|ccc|ccc|}
    \hline
    \includegraphics[scale=0.2]{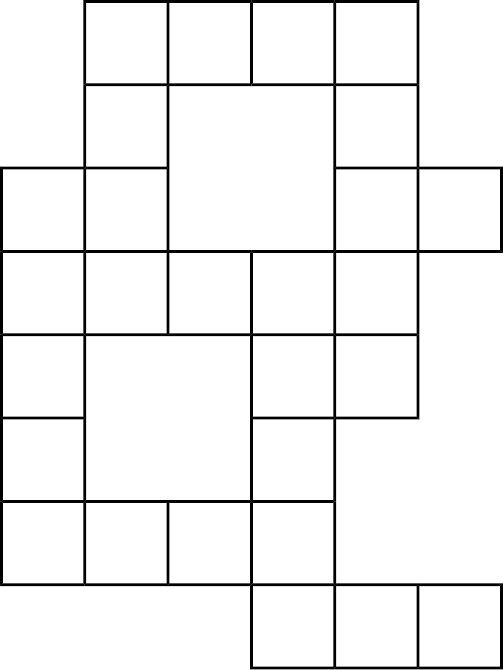} &
    \hspace{-0.2in}
    \includegraphics[scale=0.15]{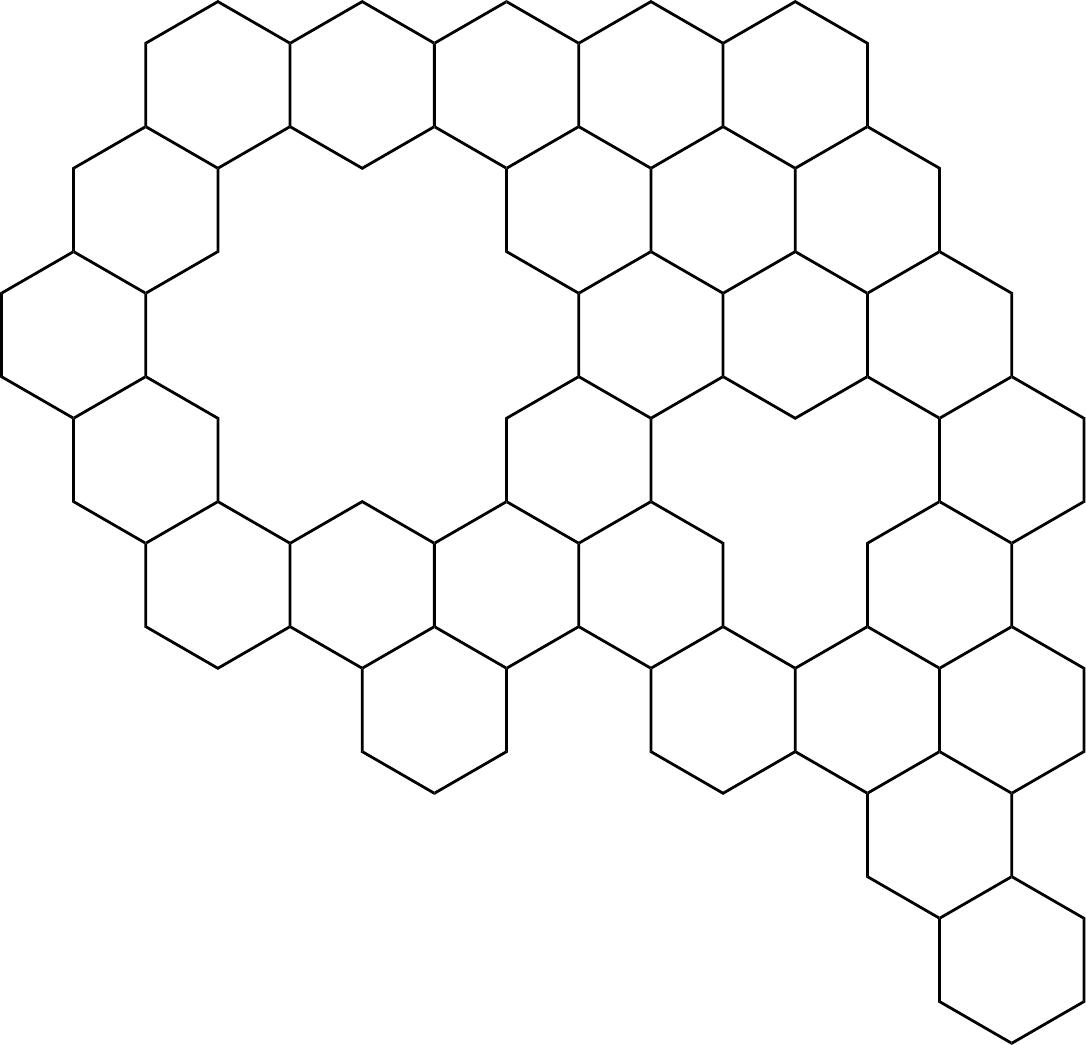} &
    \hspace{-0.2in}
    \includegraphics[scale=0.25]{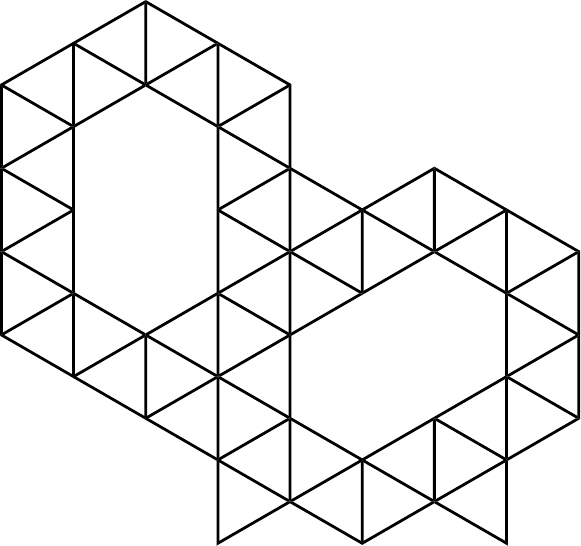} &
    \includegraphics[scale=0.2]{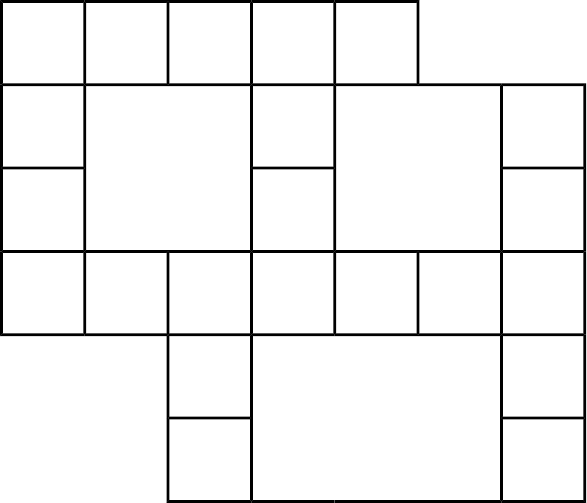} &
    \hspace{-0.1in}
    \includegraphics[scale=0.15]{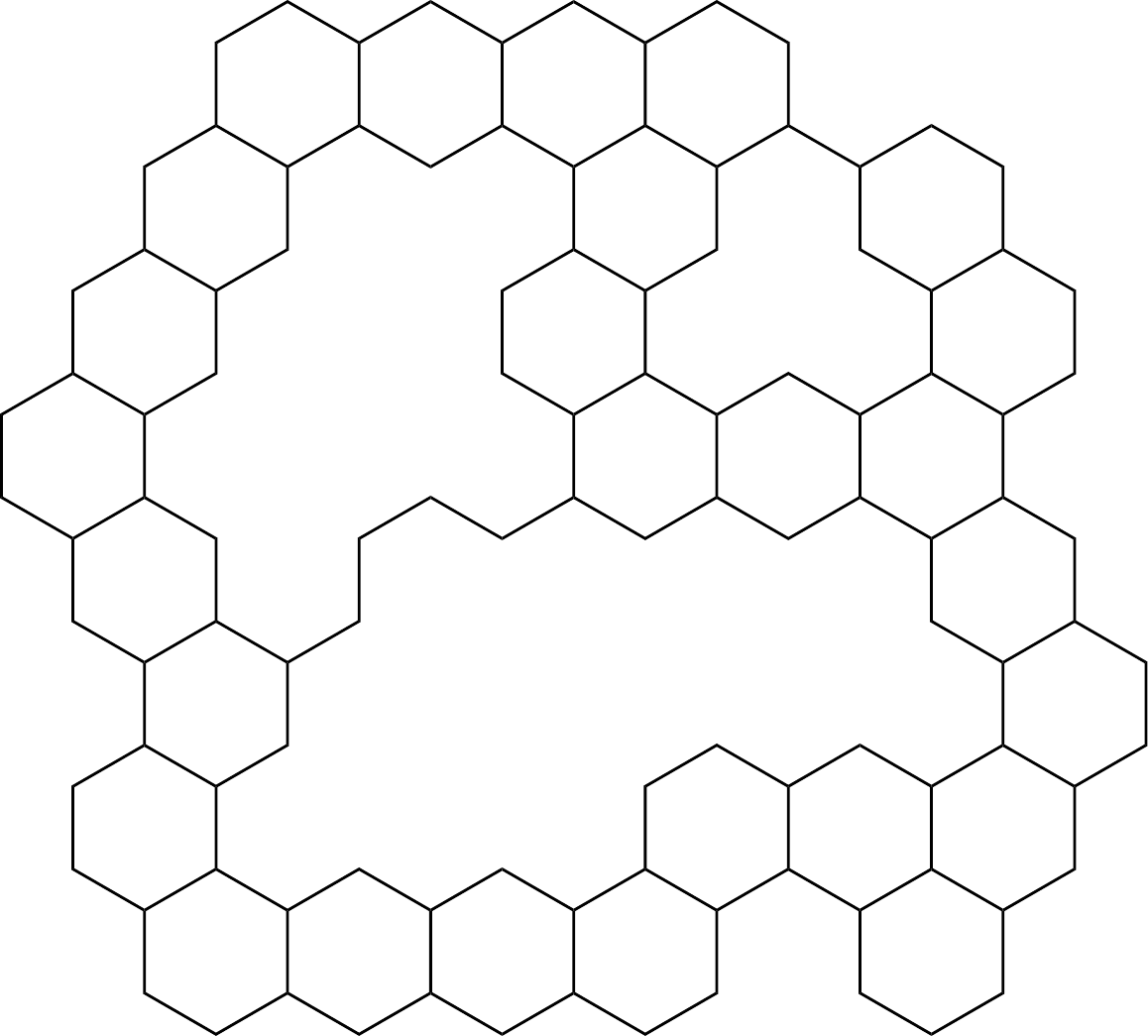} &
    \hspace{-0.1in}
    \includegraphics[scale=0.25]{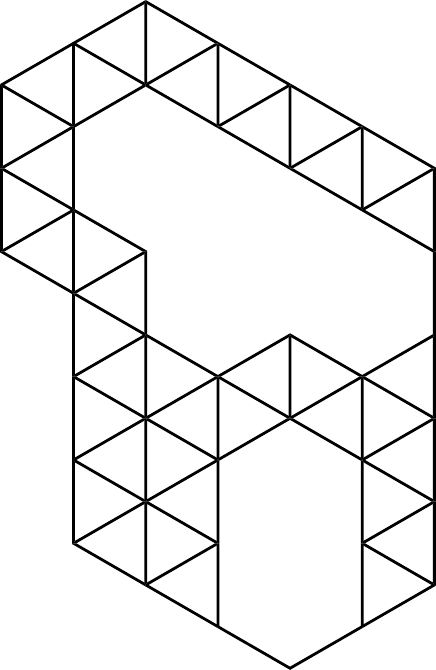} \\
    \multicolumn{3}{|c|}{\bf polygonal} &
    \multicolumn{3}{|c|}{\bf not polygonal} \\
    \hline
  %
    \hline
    \includegraphics[scale=0.2]{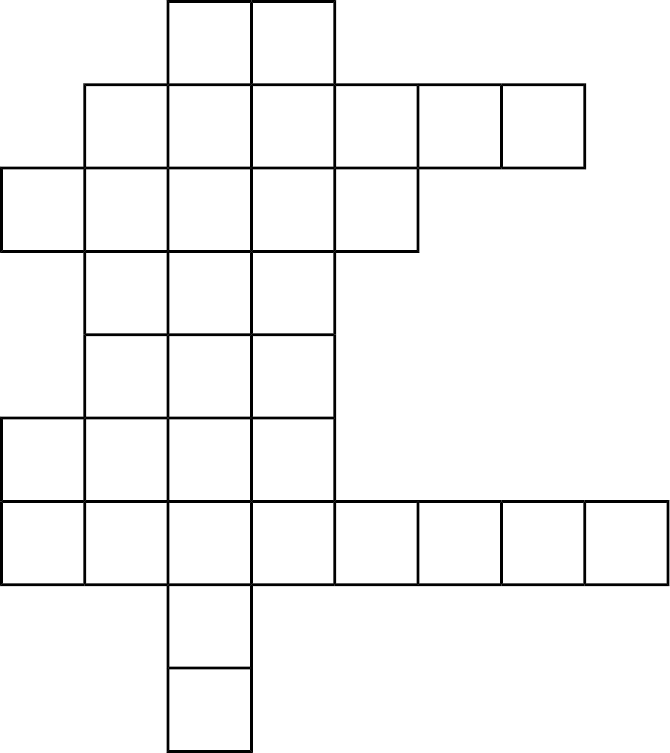} &
    \hspace{-0.1in}
    \includegraphics[scale=0.15]{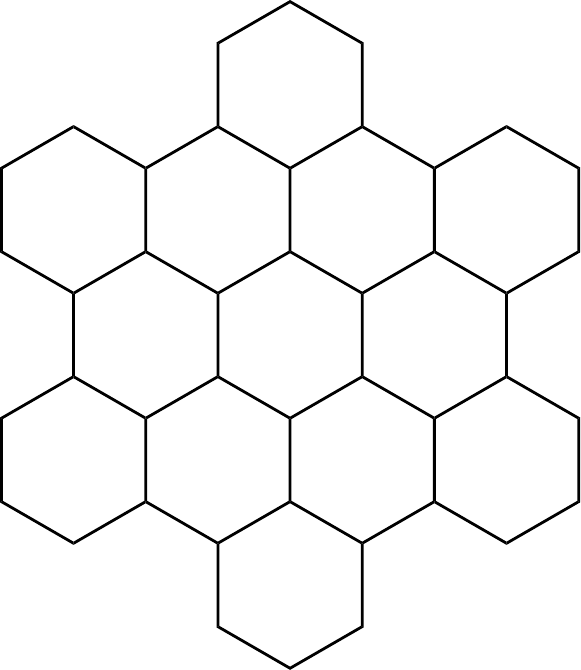} &
    \hspace{-0.1in}
    \includegraphics[scale=0.25]{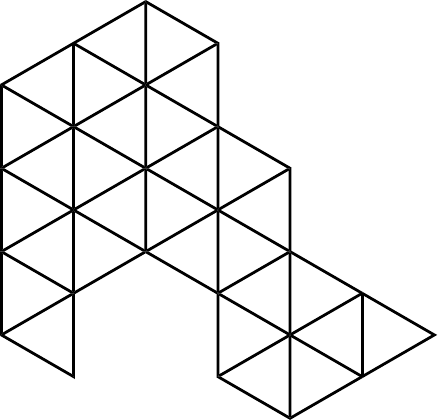} &
    \includegraphics[scale=0.2]{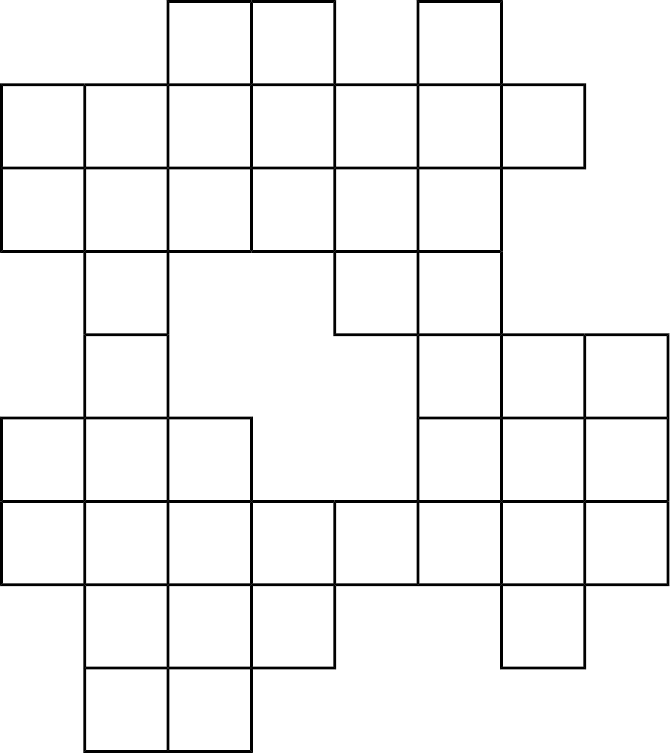} &
    \hspace{-0.1in}
    \includegraphics[scale=0.15]{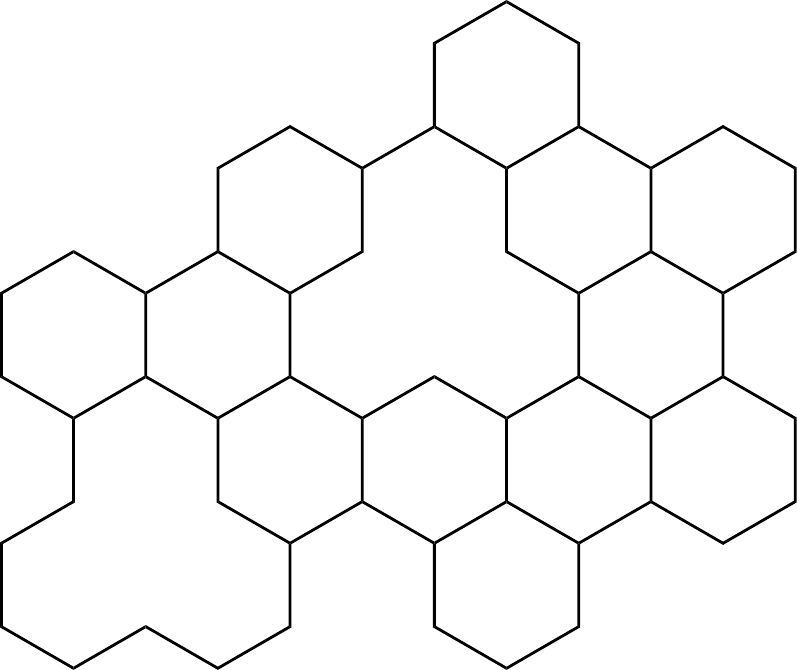} &
    \hspace{-0.1in}
    \includegraphics[scale=0.25]{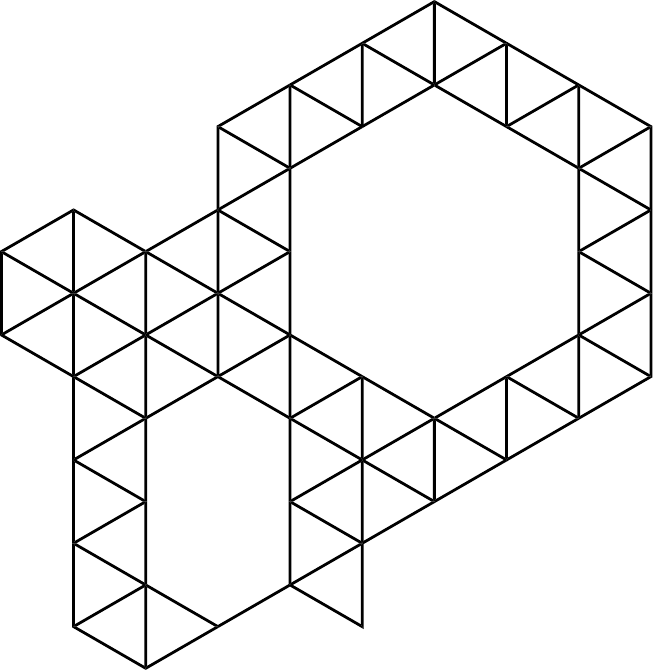} \\
    \multicolumn{3}{|c|}{\bf solid} &
    \multicolumn{3}{|c|}{\bf not solid} \\
    \hline
  %
    \hline
    \includegraphics[scale=0.2]{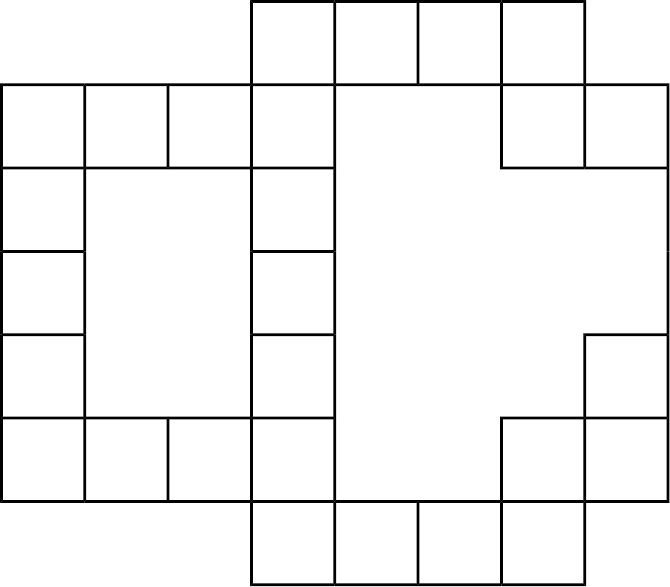} &
    \hspace{-0.15in}
    \includegraphics[scale=0.15]{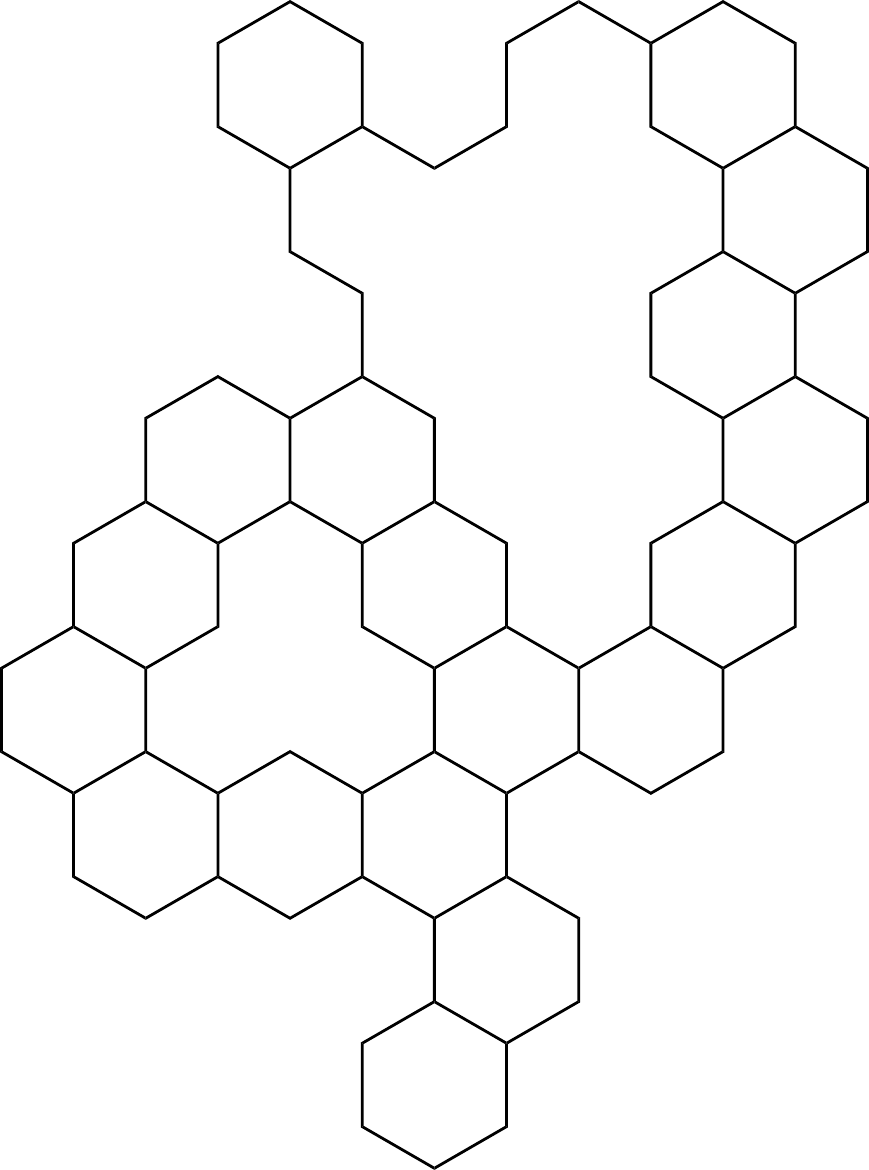} &
    \hspace{-0.25in}
    \includegraphics[scale=0.25]{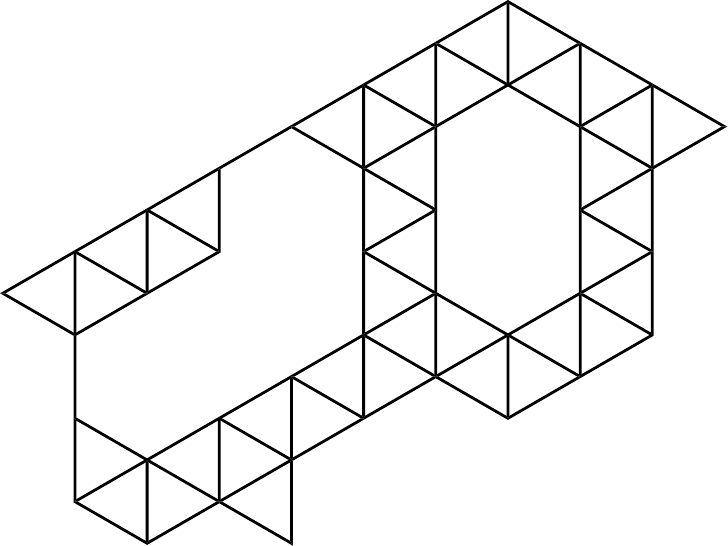} &
    \includegraphics[scale=0.2]{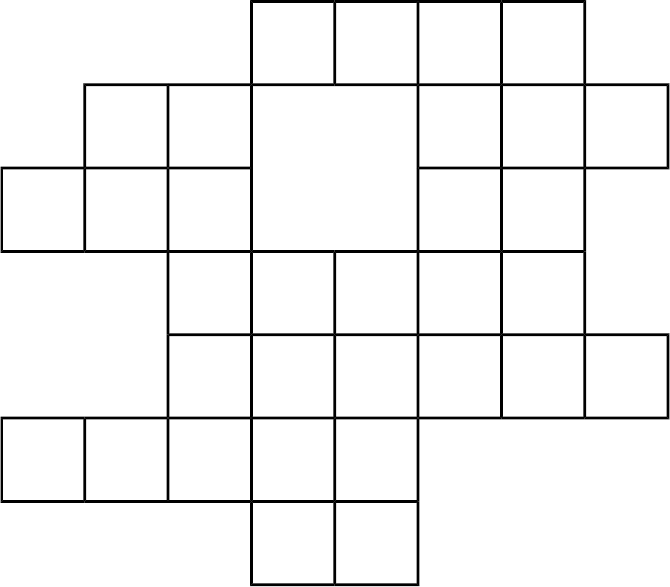} &
    \hspace{-0.2in}
    \includegraphics[scale=0.15]{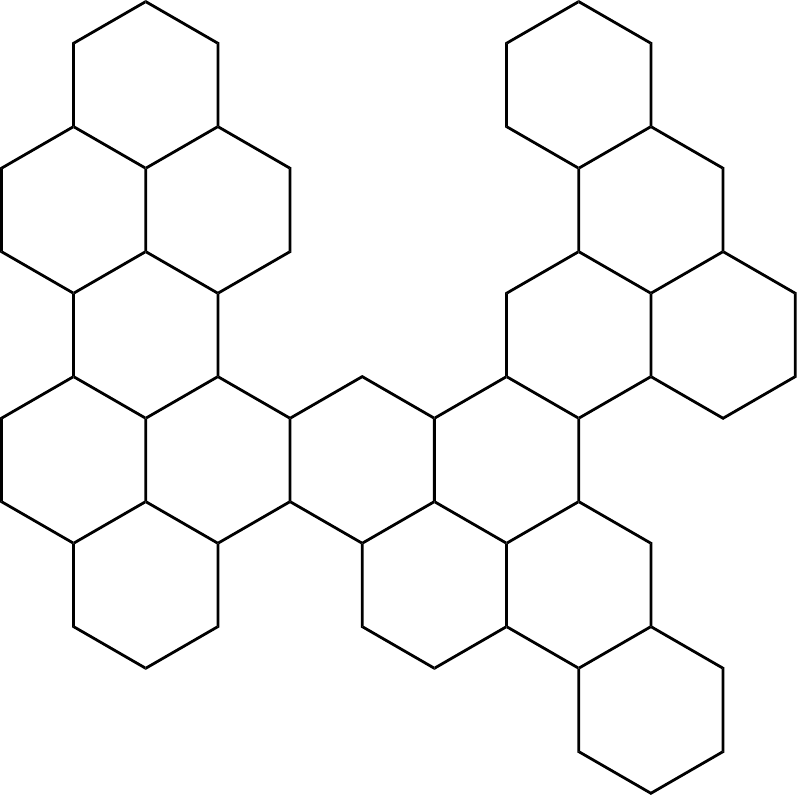} &
    \hspace{-0.25in}
    \includegraphics[scale=0.25]{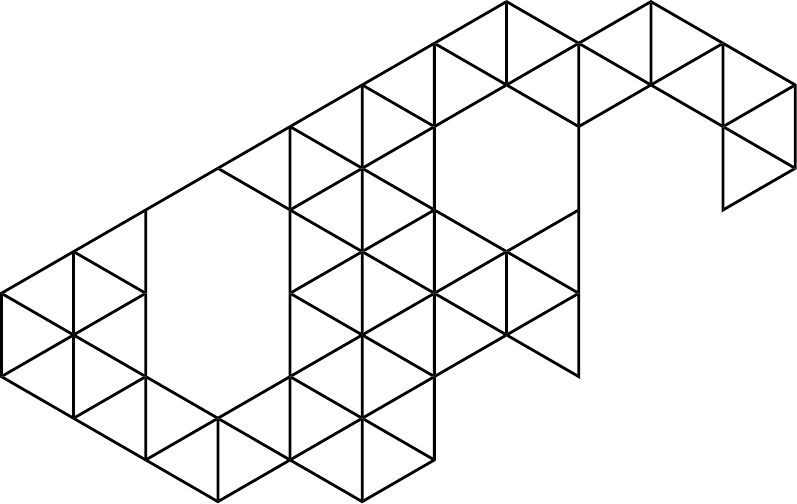} \\
    \multicolumn{3}{|c|}{\bf thin} &
    \multicolumn{3}{|c|}{\bf not thin} \\
    \hline
  \end{tabular}
  \caption{Examples of grid graphs that are or are not polygonal, solid, or thin}
  \label{fig:graphs}
\end{figure}

\begin{definition}
A grid graph is called \defn{solid} if it has no holes, or equivalently if every bounded face is a pixel. 
(See Figure~\ref{fig:graphs}, middle.)
\end{definition}

\begin{definition}
A grid graph is called \defn{thin} if every vertex in the graph is a boundary vertex. Note that a thin grid graph need not be polygonal.
(See Figure~\ref{fig:graphs}, bottom.)
\end{definition}

Now that we have defined all of the relevant terms, we can state the problems
in question:
the \prob{Hamiltonicity of [Square/Hexagonal/Triangular] [Polygonal/Thin/Polygonal Thin] Grid Graphs} problem asks whether a given [square/hexagonal/triangular] [polygonal/thin/polygonal and thin] grid graph is Hamiltonian.

\FloatBarrier
\section{Polygonal Thin Grid Graph Hamiltonicity is Easy}
\FloatBarrier
\label{sec:polygonal thin}

In this section, we show that the three polygonal thin grid graph Hamiltonicity problems are all polynomial-time solvable. This is trivial for triangular grids and easy for square grids, but is non-trivial to show for hexagonal grids.


\FloatBarrier
\subsection{Triangular Grids}
\FloatBarrier

\begin{theorem}[\cite{Arkin}]
The \prob{Hamiltonicity of Triangular Polygonal Thin Grid Graphs} problem is polynomially solvable.
\end{theorem} 

\begin{proof}
This problem is a special case of the \prob{Hamiltonicity of Triangular Polygonal Grid Graphs} problem, which was shown to be polynomially solvable in \cite{Arkin}.
\end{proof}

\FloatBarrier
\subsection{Square Grids}
\FloatBarrier

We prove below that 

\begin{theorem}
\label{theorem:polygonal_thin_square}
Every polygonal thin square grid graph is Hamiltonian. 
\end{theorem}
and therefore conclude that

\begin{corollary}
There exists a polynomial time algorithm which decides the \prob{Hamiltonicity of Square Polygonal Thin Grid Graphs} problem (the ``always accept'' algorithm). 
\end{corollary}

To prove Theorem~\ref{theorem:polygonal_thin_square}, we will provide a polynomial-time algorithm for finding a Hamiltonian cycle in a polygonal thin square grid graph, prove that if the algorithm produces an output then the output is a Hamiltonian cycle, and prove that following the algorithm is always possible.

First, the algorithm: Suppose the set of pixels in input graph $G$ is $P$. Initialize $S$ to be the empty set. Then repeat the following step until $P-S$ contains no cycles of pixels: identify a cycle of pixels in $P-S$, find a pixel $p$ in this cycle such that exactly two pixels in $G$ neighbor $p$ and the two neighboring pixels are on opposite sides of $p$, and add $p$ to $S$. Once this loop is finished, let $T = P-S$ be the set of pixels in $G$ but not in $S$. Treating $T$ as a region, output the boundary of that region as a Hamiltonian cycle.

Clearly, this algorithm is a polynomial-time algorithm. The only questions are (1) whether the output is actually a Hamiltonian cycle if the algorithm succeeds and (2) whether a given cycle of pixels always contains a pixel $p$ such that exactly two pixels in $G$ neighbor $p$ and the two neighboring pixels are on opposite sides of $p$. We prove that the answer to both these questions is ``yes'' below:

\begin{lemma}
Provided the given algorithm succeeds at each step on input $G$, the final output will be a Hamiltonian cycle in $G$.
\end{lemma}

\begin{proof}
Since $G$ is connected, the pixels in $P$ are connected as well. At every step of the algorithm, $P-S$ remains connected since the only pixel added to $S$ (and therefore removed from $P-S$) is a pixel in a cycle with no neighbors outside the cycle. Furthermore, the final value of $P-S$ (also known as $T$) will be acyclic since that is the terminating condition of the loop. Thus $T$ is connected and acyclic. In other words $T$ is a tree of pixels. As a result, the region defined by $T$ is connected and hole-free. Therefore the boundary of $T$ is one cycle. All that is left to show is that every vertex in $G$ is on this boundary.

Consider any vertex $v$ in $G$. $G$ is a thin grid graph, so every vertex, including $v$, is on the boundary of $G$. Then provided $v$ is adjacent to some pixel in $T$, we also have that $v$ is on the boundary of $T$. Thus we need to show that every vertex is adjacent to at least one pixel in $T$. 

Consider any pair of adjacent pixels $p_1$ and $p_2$ such that each of the two pixels has exactly two neighbors. As soon as one of these pixels is added to $S$ (if this ever occurs), the other will forevermore have at most one neighbor in $P-S$. As a result, this second pixel will never be in a cycle of pixels in $P-S$. Then this second pixel will never itself be added to $S$, or in other words at most one of $p_1$ and $p_2$ will be added to $S$. Thus, the final value of the set $S$ will contain no two adjacent pixels, or in other words every pixel adjacent to a pixel in $S$ will be in $T$.

But if $S$ contains pixel $p$ then every vertex adjacent to $p$ is also adjacent to one of the two neighbors of $p$ (since the two neighbors must be on opposite sides of $p$). Since these neighbors are in $T$, we see that every vertex adjacent to a pixel in $S$ is also adjacent to a pixel in $T$. 

Since the graph is polygonal, every vertex in the graph is adjacent to some pixel: either a pixel in $T$ or a pixel in $S$. In either case, we can conclude that the vertex is adjacent to a pixel in $T$, and therefore, as argued above, the boundary of $T$ is a Hamiltonian cycle in~$G$.
\end{proof}

\begin{lemma}
For any cycle of pixels $C$ in a polygonal thin grid graph $G$, there exists a pixel $p$ in that cycle such that exactly two pixels in $G$ neighbor $p$ and the two neighboring pixels are on opposite sides of $p$. 
\end{lemma}

\begin{proof}
\begin{figure}[!htbp]
  \centering
  \def\scale{.8}
  \begin{subfigure}[b]{1.6in}
    \centering
    \includegraphics[scale=\scale]{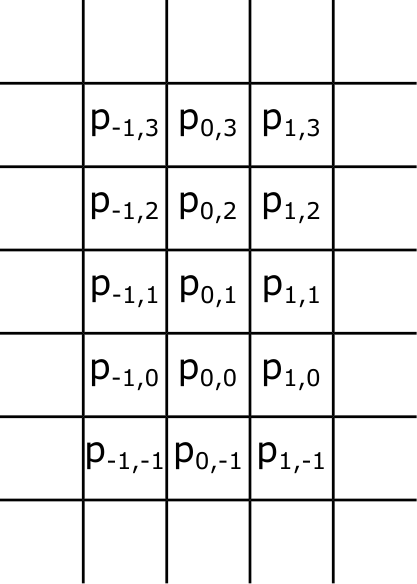}
    \caption{The naming scheme given to the pixels in the plane.\label{fig:tiling1}}
  \end{subfigure}\hfill
  \begin{subfigure}[b]{1.6in}
    \centering
    \includegraphics[scale=\scale]{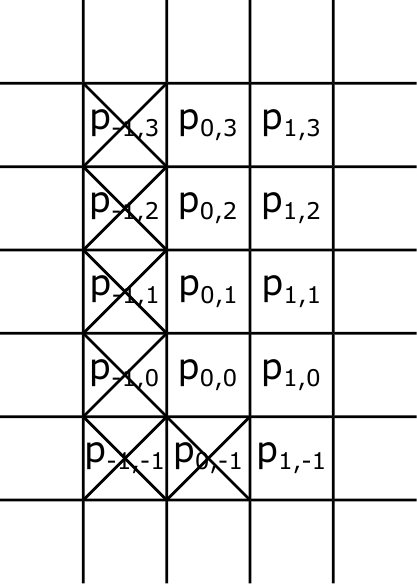}
    \caption{Pixels that cannot be in $C$ by definition of $p_{0,0}$ are crossed out.\label{fig:tiling2}}
  \end{subfigure}\hfill
  \begin{subfigure}[b]{1.6in}
    \centering
    \includegraphics[scale=\scale]{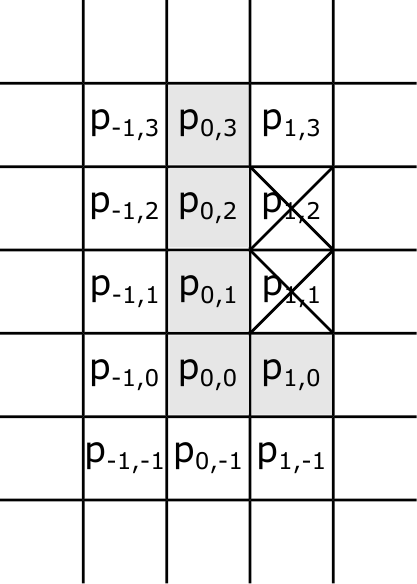}
    \caption{Pixels that must be in $C$ are shaded gray and those not in $G$ are crossed out.\label{fig:tiling3}}
  \end{subfigure}
  \caption{}
\end{figure}

Consider the leftmost column of pixels which contains any pixels in $C$ and let $p_{0,0}$ be the bottom-most pixel of $C$ in this column. Assign $x$ and $y$ coordinates to the pixels in $\mathbb{Z}_{\square}$ so that $p_{0,0}$ has coordinates $(0,0)$ and the coordinates increase as we go up and to the right. See Figure~\ref{fig:tiling1}.

By definition of $p_{0,0}$, we know that $p_{-1,i}$ is not a pixel in $C$ for any $i$ and neither is $p_{0,-1}$. These pixels are crossed out in Figure~\ref{fig:tiling2}.

But $C$ is a cycle so $p_{0,0}$ must have exactly two neighbors in $C$. Therefore $p_{1,0}$ and $p_{0,1}$ must both be in $C$. Then in order for $G$ to be thin, pixel $p_{1,1}$ cannot be in $G$ (nor in $C$). $p_{0,1}$ is a pixel in $C$ and therefore must have two neighbors in $C$. Since neither $p_{-1,1}$ nor $p_{1,1}$ are pixels in $C$ we can conclude that these two neighbors must be $p_{0,0}$ and $p_{0,2}$. In particular, $p_{0,2}$ must be a pixel in $C$.

Suppose for the sake of contradiction that $p_{1,2}$ is a pixel in $G$. Then $p_{1,2}$, $p_{0,2}$, $p_{0,1}$, $p_{0,0}$, and $p_{1,0}$ are all pixels in $G$. As a result, all four vertices on the boundary of pixel $p_{1,1}$ are in $G$, and so since $G$ is an induced subgraph, the edges between these vertices are in $G$ as well. As a result, we can conclude that pixel $p_{1,1}$ is in $G$, which is a contradiction. Thus $p_{1,2}$ is not a pixel in $G$.

Pixel $p_{0,2}$ is in $C$ and therefore must have two neighbors in $C$. Since neither $p_{-1,2}$ nor $p_{1,2}$ is in $C$ we can conclude that these two neighbors must be $p_{0,1}$ and $p_{0,3}$. In particular, $p_{0,3}$ must be in $C$.

We have shown that $p_{0,0}$, $p_{0,1}$, $p_{0,2}$, and $p_{0,3}$ are all pixels in $C$ and that $p_{1,1}$ and $p_{1,2}$ are not pixels in $G$. See Figure~\ref{fig:tiling3}. Since $G$ is thin, either $p_{-1,1}$ or $p_{-1,2}$ must be a pixel not in $G$. Then for some $i \in \{1,2\}$ we have that $p_{0, i}$ is a pixel in $C$, $p_{0, i\pm1}$ are pixels in $C$, and $p_{\pm 1, i}$ are pixels not in $G$. 

That pixel $p_{1, i}$ is the pixel $p$ we wished to find: a pixel in $C$ such that exactly two pixels in $G$ neighbor $p$ and the two neighboring pixels are on opposite sides of $p$. 
\end{proof}

\FloatBarrier
\subsection{Hexagonal Grids}
\FloatBarrier
\label{hex_grids_subsection}

Consider the following problem:

\begin{restatable}{problem}{trvbprob}
The \prob{Tree-Residue Vertex-Breaking} problem asks for a given multigraph $G$ in which every vertex is labeled as ``breakable'' or ``unbreakable'' whether there exists a subset of the breakable vertices such that ``breaking'' those vertices results in a tree.

Here the operation of \emph{breaking} a vertex in a multigraph (shown in Figure~\ref{fig:vertex_breaking_example}) results in a new multigraph by removing the vertex, adding a number of new vertices equal to the degree of the vertex in the original multigraph, and connecting these new vertices to the neighbors of the vertex in a one-to-one manner. 
\end{restatable}

\begin{figure}[!htbp]
  \centering
  \def\scale{0.35}
  \hspace{1in}
  \begin{subfigure}[b]{0.8in}
    \centering
    \includegraphics[scale=\scale]{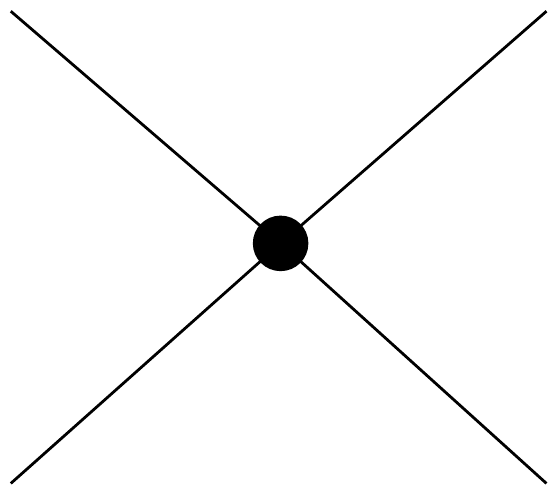}
    \caption{Before.}
  \end{subfigure}\hfill
  \begin{subfigure}[b]{.8in}
    \centering
    \includegraphics[scale=\scale]{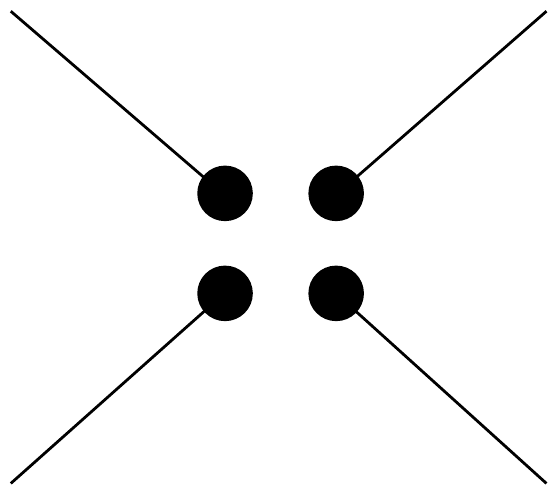}
    \caption{After.}
  \end{subfigure}
  \hspace{1in}
  \caption{An example of breaking a vertex}
  \label{fig:vertex_breaking_example}
\end{figure}

This problem and its variants were studied in \cite{trvb}:

\begin{theorem} {\rm \cite{trvb}}
The \prob{Max-Degree-$3$ Tree-Residue Vertex-Breaking} problem asks the same question as the \prob{Tree-Residue Vertex-Breaking} problem, but restricts the inputs to be graphs whose vertices each have degree at most~$3$. The \prob{Max-Degree-$3$ Tree-Residue Vertex-Breaking} problem is polynomial-time solvable.
\end{theorem}

In this section, we will show that

\begin{theorem}
There exists a polynomial-time reduction from the \prob{Hamiltonicity of Hexagonal Polygonal Thin Grid Graphs} problem to the \prob{Max-Degree-$3$ Tree-Residue Vertex-Breaking} problem.
\label{thm:hex_polygonal_thin_overall}
\end{theorem}
and therefore that 

\begin{corollary}
The \prob{Hamiltonicity of Hexagonal Polygonal Thin Grid Graphs} problem is polynomial-time solvable. 
\end{corollary}

We prove Theorem~\ref{thm:hex_polygonal_thin_overall} with the following reduction: On input a hexagonal polygonal thin grid graph $G$, construct the graph $G'$ whose vertices are pixels and whose edges connect adjacent vertices. Label a vertex in $G'$ as breakable if it has degree-$3$. Otherwise label the vertex unbreakable. Output the resulting labeled graph.

To prove the correctness of this reduction, we first consider the behavior of a Hamiltonian cycle in the local vicinity of a pixel, then narrow down the possibilities using non-local constraints, and finally use the global constraints imposed by the existence of a Hamiltonian cycle to characterize the hexagonal polygonal thin grid graphs with Hamiltonian cycles. We will show using this characterization that a hexagonal polygonal thin grid graph is Hamiltonian if and only if the corresponding reduction output is a ``yes'' instance of \prob{Max-Degree-$3$ Tree-Residue Vertex-Breaking}.

\begin{lemma}
In a hexagonal polygonal thin grid graph, the only possible behaviors for a Hamiltonian cycle near a given pixel in the graph look like one of the diagrams in Figures \ref{fig:hex_zero}, \ref{fig:hex_one}, \ref{fig:hex_two}, or \ref{fig:hex_three} depending on the number of other pixels adjacent to the pixel in question.
\end{lemma}

\begin{figure}[!htbp]
  \centering
  \def\scale{0.35}
  \begin{subfigure}[b]{0.8in}
    \centering
    \includegraphics[scale=\scale]{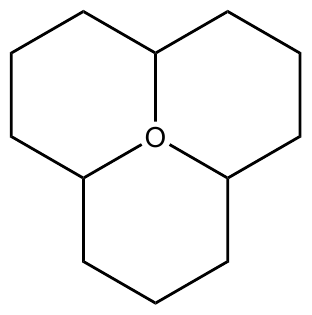}
    \caption{Not thin\label{fig:three_hexes}}
  \end{subfigure}\hfill
  \begin{subfigure}[b]{1in}
    \centering
    \includegraphics[scale=\scale]{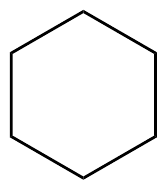}
    \caption{Zero neighbors\label{fig:hex_zero}}
  \end{subfigure}\hfill
  \begin{subfigure}[b]{1in}
    \centering
    \includegraphics[scale=\scale]{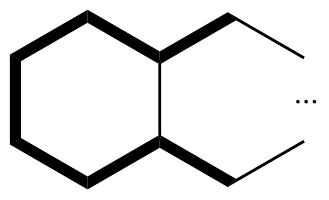}
    \caption{One neighbor\label{fig:hex_one}}
  \end{subfigure}\hfill
  \begin{subfigure}[b]{2.5in}
    \centering
    \includegraphics[scale=\scale]{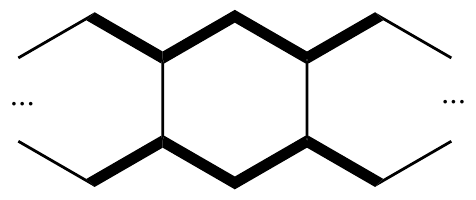}\hfill
    \includegraphics[scale=\scale]{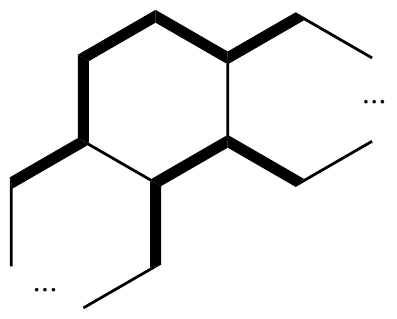}\hfill
    \includegraphics[scale=\scale]{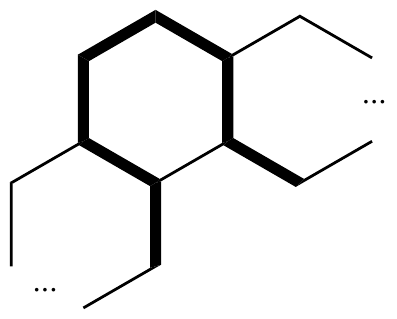}
    \caption{Two neighbors\label{fig:hex_two}}
  \end{subfigure}
  \begin{subfigure}[b]{\linewidth}
    \centering
    \includegraphics[scale=\scale]{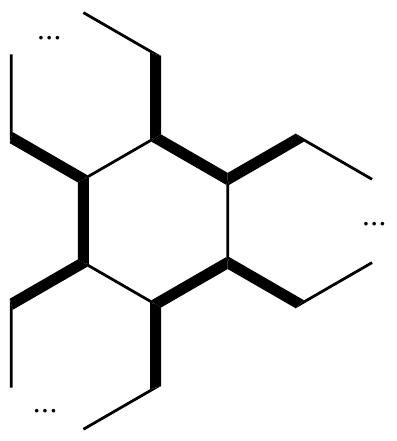}\hfill
    \includegraphics[scale=\scale]{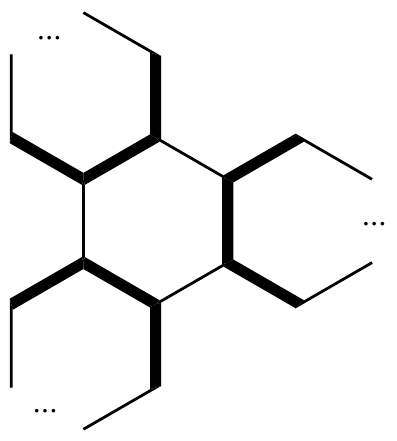}\hfill
    \includegraphics[scale=\scale]{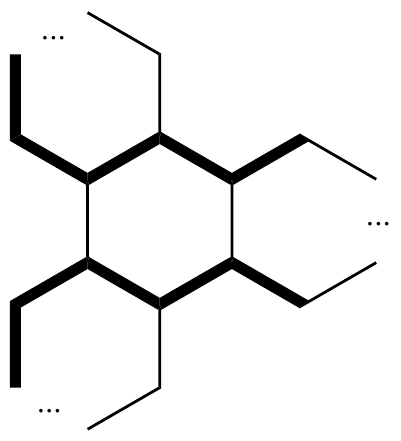}\hfill
    \includegraphics[scale=\scale]{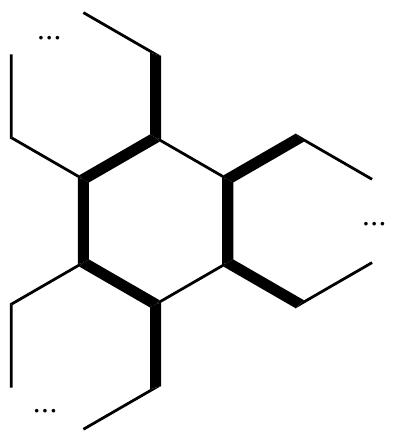}\hfill
    \includegraphics[scale=\scale]{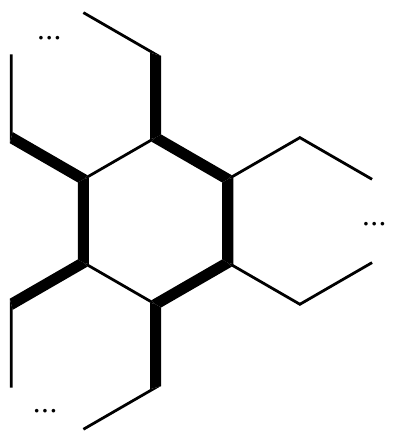}\hfill
    \includegraphics[scale=\scale]{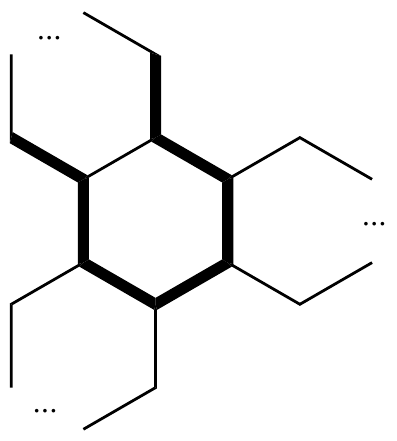}\hfill
    \includegraphics[scale=\scale]{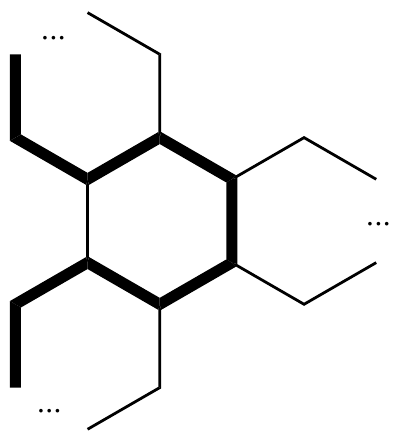}
    \caption{Three neighbors\label{fig:hex_three}}
  \end{subfigure}
  \caption{Local constraints on hexagons}
\end{figure}

\begin{proof}
Consider the pattern of pixels in Figure~\ref{fig:three_hexes}. In any grid graph containing three pixels in this arrangement, the circled vertex is not on the boundary. Thus, because every vertex must be on the boundary of a thin grid graph, the three-pixel pattern does not occur in any hexagonal thin grid graph. 

Next consider a single pixel. It can have up to 6 neighboring pixels, but in order to avoid the three pixel arrangement from Figure~\ref{fig:three_hexes}, it will have at most 3 neighbors.

If a pixel has zero neighbors (i.e., Figure~\ref{fig:hex_zero}), it is a single connected component of the graph. Because the graphs we consider are connected, that means that the pixel is the entire graph. In that case, there is a Hamiltonian cycle (consisting of the whole graph). Since in this case we can easily solve the Hamiltonicity problem, we restrict our attention to other cases.

If a pixel has exactly one neighbor, the cycle must pass through it as shown in Figure~\ref{fig:hex_one} (up to rotation).

If a pixel has exactly two neighbors, they can be arranged (up to rotation) in two ways. If the two neighboring pixels are opposite, the cycle must pass through the pixel as shown in Figure~\ref{fig:hex_two} (left). In the other arrangement, there are two possibilities as shown in Figure~\ref{fig:hex_two} (right).

Finally, there is exactly one way to arrange a pixel with three neighboring pixels, and the cycle can pass through this type of pixel in seven different ways. This is shown in Figure~\ref{fig:hex_three}.

The different possibilities of how to arrange the cycle in a given diagram were computed by exhaustive search subject to the constraints that (1) every vertex must be in the cycle and (2) exactly two edges next to each vertex must be in the cycle.
\end{proof}

\begin{lemma}
In a hexagonal polygonal thin grid graph, the only possible behaviors for a Hamiltonian cycle near a given pixel look like one of the diagrams in Figure~\ref{fig:restricted_possibilities}. Furthermore, the final situation in Figure~\ref{fig:restricted_possibilities} necessarily leads to the situation shown in the rightmost part of Figure~\ref{fig:weird}.
\end{lemma}

\begin{figure}[!htbp]
  \centering
  \def\scale{0.28}
  \def\sscale{0.22}
  \def\gap{\hspace{0.0in}}
  \subcaptionbox{\label{fig:bad_1}}
   {\includegraphics[scale=\scale]{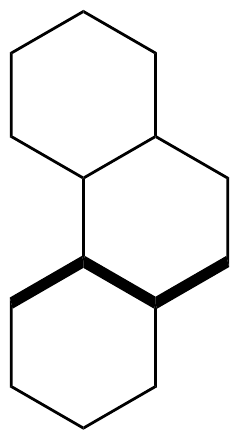}~~
    \includegraphics[scale=\scale]{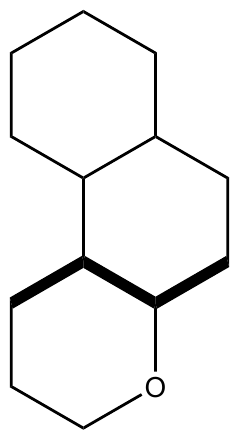}~~
    \includegraphics[scale=\scale]{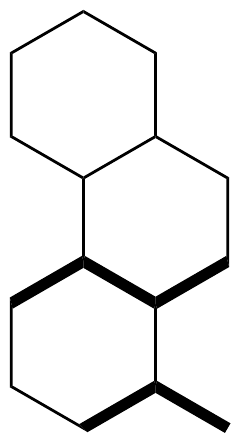}~~
    \includegraphics[scale=\scale]{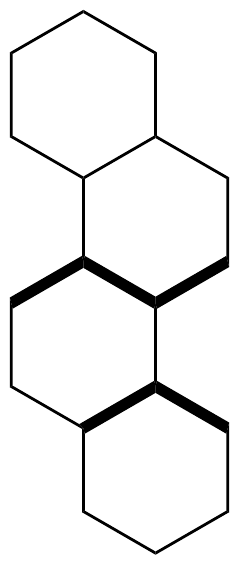}~~
    \includegraphics[scale=\scale]{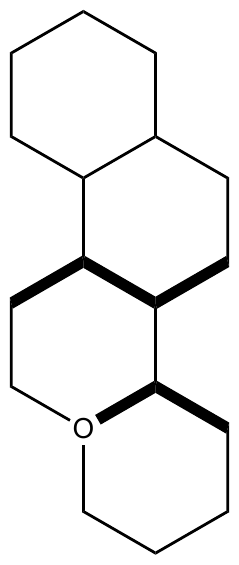}}
    
  \subcaptionbox{\label{fig:bad_2}}
   {\includegraphics[scale=\scale]{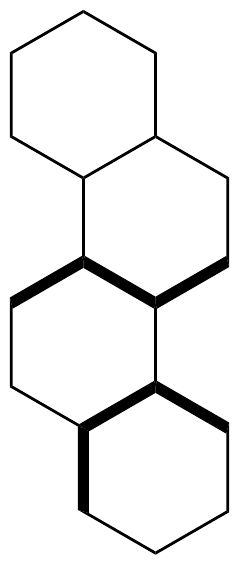}~~
    \includegraphics[scale=\scale]{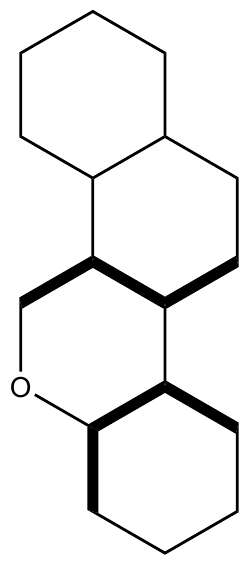}~~
    \includegraphics[scale=\scale]{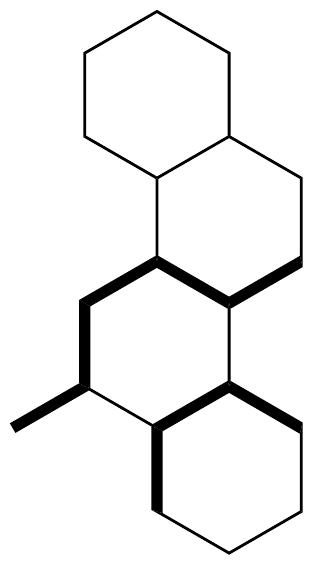}~~
    \includegraphics[scale=\scale]{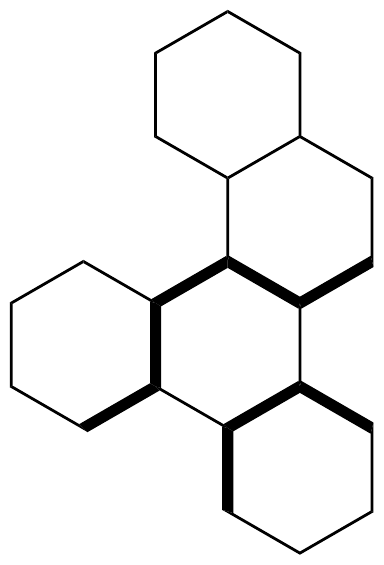}~~
    \includegraphics[scale=\scale]{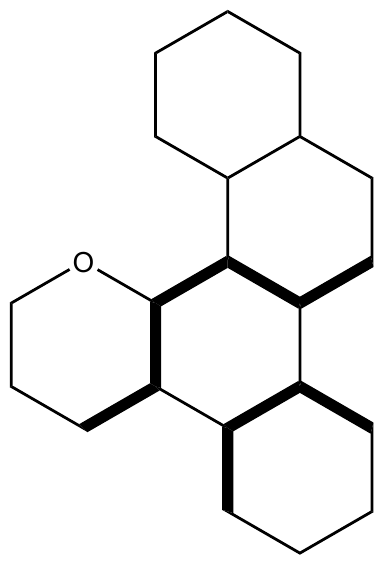}~~
    \includegraphics[scale=\scale]{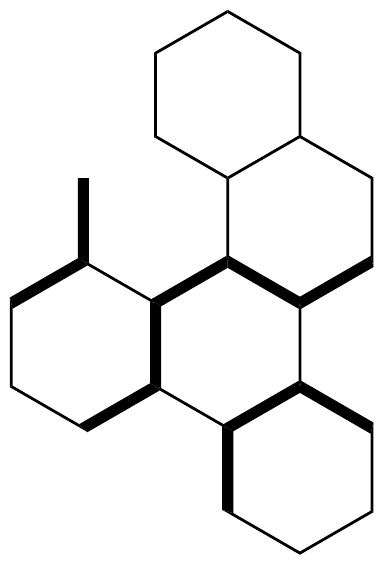}~~
    \includegraphics[scale=\scale]{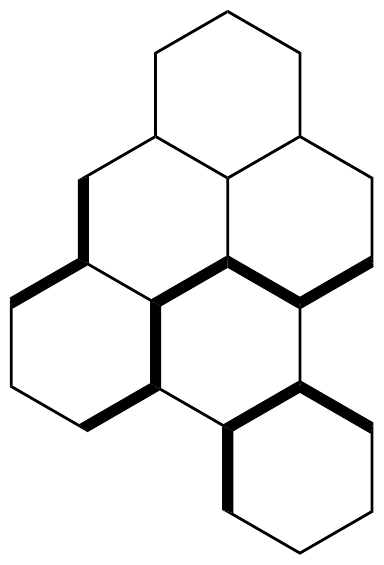}}\hfill
  \subcaptionbox{\label{fig:bad_3}}
   {\includegraphics[scale=\scale]{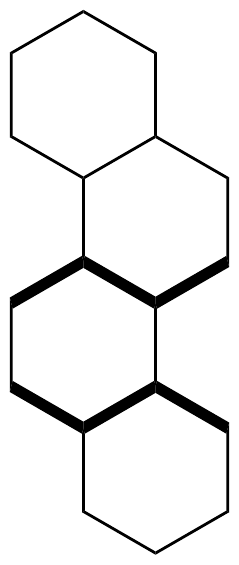}~~
    \includegraphics[scale=\scale]{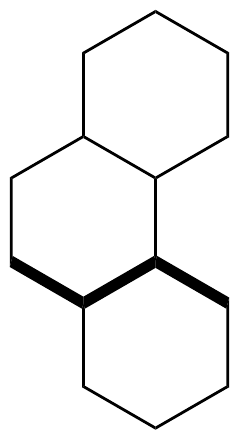}}\hfill
  \subcaptionbox{\label{fig:other_bad}}
   {\includegraphics[scale=\scale]{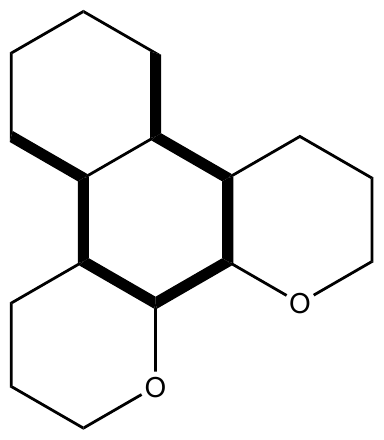}~~
    \includegraphics[scale=\scale]{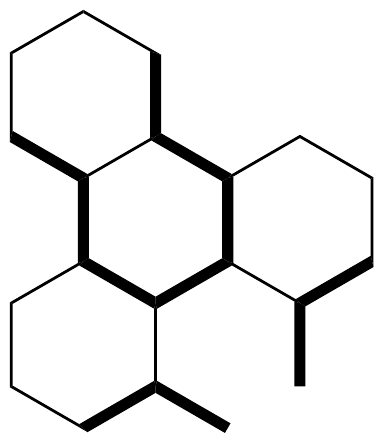}~~
    \includegraphics[scale=\scale]{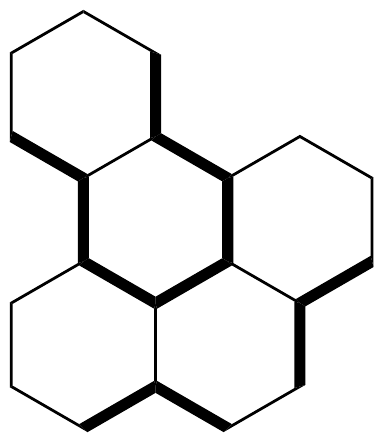}}\hfill
  \subcaptionbox{\label{fig:restricted_possibilities}}
   {\includegraphics[scale=\sscale]{images/hexagonal_polygonal_thin/hex_one}~~
    \includegraphics[scale=\sscale]{images/hexagonal_polygonal_thin/hex_two_a}~~
    \includegraphics[scale=\sscale]{images/hexagonal_polygonal_thin/hex_two_b_1}~~
    \includegraphics[scale=\sscale]{images/hexagonal_polygonal_thin/hex_three_1}~~
    \includegraphics[scale=\sscale]{images/hexagonal_polygonal_thin/hex_three_2}~~
    \includegraphics[scale=\sscale]{images/hexagonal_polygonal_thin/hex_three_5}}
   
  \subcaptionbox{\label{fig:weird}}
   {\includegraphics[scale=\scale]{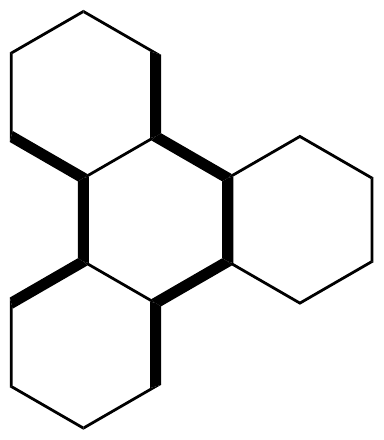}~~
   \includegraphics[scale=\scale]{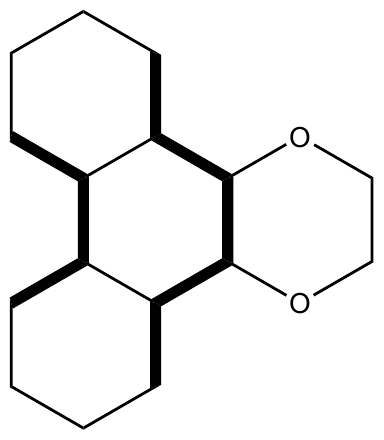}~~
   \includegraphics[scale=\scale]{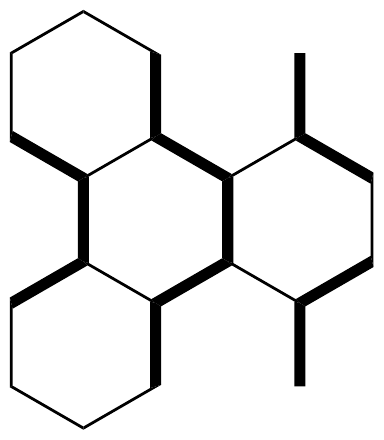}~~
   \includegraphics[scale=\scale]{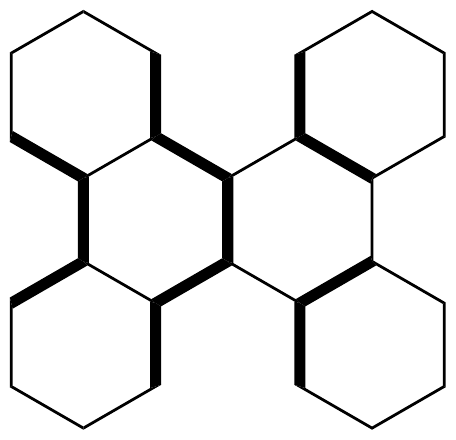}~~
   \includegraphics[scale=\scale]{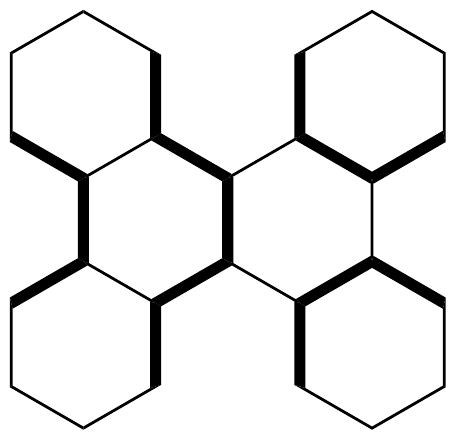}}
  \caption{Less local constraints on hexagons}
\end{figure}

\begin{proof}
Suppose that we have three pixels arranged as shown in the leftmost part of Figure~\ref{fig:bad_1} with the three bold edges definitely included in the Hamiltonian Cycle. Because the circled vertex in the second part of the figure has to have two edges, we can conclude that the situation must be as shown in the third part of the figure. But the bottom right edge must be part of a pixel, and because every vertex must be on the boundary, only one of the two possible pixels will work. This yields the situation in the fourth part of the figure. Consider the circled vertex in the fifth part of the figure. That vertex must have two edges in the Hamiltonian cycle touching it. One is already accounted for, so we just have to decide on the other.

Consider for the sake of contradiction that the chosen edge was as shown in the first part of Figure~\ref{fig:bad_2}. Then because of the circled vertex in the second part of the figure, the situation must be as shown in the third part of the figure. Again, every edge must be part of a pixel, and again only one of the two possible pixels would yield a thin graph. Thus we arrive at the situation in the fourth part of the figure. Because of the circled vertex in the fifth part of the figure, we arrive at the situation in the sixth. But then the dotted edge in the seventh figure must also be in the graph (because no two adjacent vertices can exist in the graph without the edge between them also being in the graph). This, however, yields a graph that is not thin. Thus we have a contradiction.

Therefore, what must actually have happened is shown in the left side of Figure~\ref{fig:bad_3}. Looking only at the bottom three pixels in this new situation and ignoring some of the bold lines, we have the situation in the right half of the figure. Note, however, that this situation is identical (up to reflection and translation down) to the situation at the start of Figure~\ref{fig:bad_1}. Thus, having the pattern in the first part of Figure~\ref{fig:bad_1} necessitates another copy of the same pattern lower down (and flipped). That in turn necessitates another copy lower down, which necessitates another copy lower down, etc... In other words, no finite graph can contain the pattern shown at the start of Figure~\ref{fig:bad_1}.  

As a result of this, many of the ``possible'' local solutions at a pixel are actually not allowed. 

We can restrict the possibilities even more, however. Consider the penultimate scenario from Figure~\ref{fig:hex_three}. Under this scenario, consider the circled vertices in the first part of Figure~\ref{fig:other_bad}. The constraints at those vertices lead to the scenario in the center part of that figure, which in turn leads to the existence of the edge added in the final part of the figure. This contradicts the fact that the grid graph must be thin, so the actual list of possible local solutions does not include the penultimate scenario from Figure~\ref{fig:hex_three}. 

Then the restricted list of possibilities is as shown in Figure~\ref{fig:restricted_possibilities} where we include only those local solutions which omit the pattern at the start of Figure~\ref{fig:bad_1} and where we exclude the penultimate scenario from Figure~\ref{fig:hex_three}.

Finally, we wish to show that the last situation listed in Figure~\ref{fig:restricted_possibilities} directly leads to the situation shown in the last part of Figure~\ref{fig:weird}. The first part of Figure~\ref{fig:weird} is exactly the final situation listed in Figure~\ref{fig:restricted_possibilities}. Consider the circled vertices in the second part of Figure~\ref{fig:weird}. The constraints at these vertices lead to the situation shown in the third part of Figure~\ref{fig:weird}. Next, due to the thin and polygonal properties, the pixels added in the penultimate part of Figure~\ref{fig:weird} must be present in the graph. Consider the rightmost pixel with three neighbors. We have a list of situations that are possible in the vicinity of a pixel with three neighbors (in Figure~\ref{fig:restricted_possibilities}), and only one of those possibilities matches the already chosen edges. Thus, we must have the situation shown in the last part of Figure~\ref{fig:weird}.
\end{proof}

\begin{lemma}
There exists a Hamiltonian cycle in a hexagonal polygonal thin grid graph if and only if there exists a tree of pixels such that every pixel with fewer than three neighbors is in the tree and such that at least one pixel out of every pair of adjacent degree-$3$ pixels is in the tree.
\end{lemma}

\begin{proof}
First suppose there exists a Hamiltonian cycle in a hexagonal polygonal thin grid graph. 

Based on the local solutions for degree-$1$ and degree-$2$ pixels, whenever one pixel with degree at most two is adjacent to another, the two pixels are on the same side of the cycle (both inside or both outside). Based on the local solutions for degree-$3$ pixels (and based on the specific situation that the final situation listed in Figure~\ref{fig:restricted_possibilities} necessarily leads to), whenever two degree-$1$ or degree-$2$ pixels are adjacent to the same degree-$3$ pixel, the two pixels are on the same side of the cycle. Finally, whenever two degree-$3$ pixels are neighbors, their other four neighbors are all degree-$1$ or degree-$2$ pixels and are all on the same side of the cycle. All together, this implies that all pixels with at most two neighbors must be on the same side of the cycle. And because the unbounded face is outside the cycle, the faces next to it---pixels---must be inside the cycle. 

We can conclude that all pixels with two or fewer neighbors are always inside the cycle. This immediately implies that all holes are outside the cycle.

Clearly, the pixels inside the Hamiltonian cycle are connected. Furthermore, they are acyclic because a cycle inside the Hamiltonian cycle would imply either that the Hamiltonian cycle contains a hole or that the graph contains an interior vertex. Thus the pixels inside the Hamiltonian cycle form a tree.

In addition, it is easy to verify that whenever two degree-$3$ pixels are adjacent, they are not both outside the Hamiltonian cycle.

Next suppose that there exists a tree of pixels such that every pixel with fewer than three neighbors is in the tree and such that at least one pixel out of every pair of adjacent degree-$3$ pixels is in the tree. Then consider the perimeter of this tree of pixels. We claim that the perimeter is a Hamiltonian cycle. Clearly, the perimeter is a cycle, so the only question is whether it touches every vertex. Every vertex is either between two degree-$3$ vertices or adjacent to at least one degree-$2$ or degree-$1$ pixel. In either case, the assumptions are sufficient to show that the vertex is on the perimeter of at least one pixel in the tree. Because the graph is thin, there are no interior vertices; in other words, every vertex on the perimeter of a pixel in the tree is also on the perimeter of the entire tree of pixels. Thus the perimeter of the tree of pixels hits every vertex and is therefore a Hamiltonian cycle.

We have now shown both directions of the desired statement.
\end{proof}

\begin{lemma}
Suppose $G$ is a hexagonal polygonal thin grid graph and $G'$ is the output of the reduction on input $G$. Then $G'$ is a ``yes'' instance of \prob{Max-Degree-$3$ Tree-Residue Vertex-Breaking} if and only if there exists a tree of pixels in $G$ such that every pixel with fewer than three neighbors is in the tree and such that at least one pixel out of every pair of adjacent degree-$3$ pixels is in the tree.
\end{lemma}

\begin{proof}
First suppose that $G'$ is a ``yes'' instance of \prob{Max-Degree-$3$ Tree-Residue Vertex-Breaking}. Then let $B$ be the set of breakable vertices in $G'$ whose breaking yields a tree. Vertices in $G'$ correspond with pixels in $G$, so define $T$ to be the set of pixels in $G$ that do not correspond to vertices in $B$.

Breaking a vertex in $G'$ corresponds to removing the vertex and then adding some number of degree-$1$ vertices. The result of breaking the vertices of $B$ in $G'$ is a tree $G'_{break}$. If instead of breaking the vertices we simply removed them, we would get a graph $G'_{remove}$ which could instead be obtained by removing the degree-$1$ vertices that were added during the breaking operation. Since removing degree-$1$ vertices from a tree yields a tree, we can conclude that $G'_{remove}$ is a tree. But $G'_{remove}$ has as vertices the pixels of $T$ and as edges the adjacencies between pixels. Thus $T$ forms a tree of pixels.

By the way the reduction was defined, every breakable vertex corresponds to a pixel in $G$ of degree-$3$. Thus $B$ contains only vertices of degree-$3$ and so $T$ contains every pixel that has fewer than three neighbors.

Consider any pair of adjacent degree-$3$ pixels. If the vertices corresponding to both pixels were in $B$ then both vertices would be broken, leading to the edge between those two vertices being disconnected from the rest of the graph. Since breaking the vertices of $B$ in $G'$ yields a tree, this is not the case. Thus, at least one of the two vertices is not in $B$, and so at least one of the pair of pixels is in $T$.

We have shown that $T$ is a tree of pixels in $G$ such that every pixel with fewer than three neighbors is in the tree and such that at least one pixel out of every pair of adjacent degree-$3$ pixels is in the tree. Thus whenever $G'$ is a ``yes'' instance of \prob{Max-Degree-$3$ Tree-Residue Vertex-Breaking} such a tree of pixels must exist.

Next we prove the other direction. Suppose that $T$ is a tree of pixels in $G$ such that every pixel with fewer than three neighbors is in the tree and such that at least one pixel out of every pair of adjacent degree-$3$ pixels is in the tree. Define $B$ to be the set of breakable vertices in $G'$ corresponding to pixels of $G$ that are not in $T$. 

Consider the graph $G'_{remove}$, which is constructed from $G'$ by removing every vertex in $B$. The vertices of this graph are the pixels in $T$ and the edges are adjacencies of pixels. Thus, since $T$ is a tree of pixels, $G'_{remove}$ is a tree. 

Also consider the graph $G'_{break}$ that results from breaking the vertices in $B$. Breaking a vertex can be accomplished by removing it and then adding some number of degree-$1$ vertices. Therefore the graph $G'_{break}$ could also be constructed from $G'_{remove}$ by adding some number of degree-$1$ vertices. If every vertex that is added has as its sole neighbor a vertex from $G'_{remove}$ then the addition of these degree-$1$ vertices maintains the tree property, allowing us to conclude that $G'_{break}$ is a tree.

Note that by the definition of $T$, no two vertices in $B$ are adjacent. Suppose $v'$ is a degree-$1$ vertex created to neighbor vertex $u \in G'$ during the breaking operation of vertex $v \in B$. Since $B$ contains no adjacent pairs of vertices, $u \not\in B$, and therefore $u$ is never broken when constructing $G'_{break}$. Thus the sole neighbor of $v'$ in $G'_{break}$ is $u$, which is a vertex from $G'_{remove}$. As stated above, applying this logic to every $v'$ allows us to assert that $G'_{break}$ is a tree, and therefore (since the maximum degree of vertices in $G'$ is at most $3$) that $G'$ is a ``yes'' instance of \prob{Max-Degree-$3$ Tree-Residue Vertex-Breaking}.

We have shown both directions, thus proving the desired equivalence.
\end{proof}

Putting the last two lemmas together, we conclude that the given reduction is correct.

\FloatBarrier
\section{Hexagonal Thin Grid Graph Hamiltonicity is NP-complete}
\FloatBarrier
\label{sec:hex thin}

In this section, we show that the \prob{Hexagonal Thin Grid Graph Hamiltonicity} problem is NP-complete. Membership in NP is trivial, while NP-hardness follows via reduction from \prob{$6$-Regular Breakable Planar Tree-Residue Vertex-Breaking}.

Recall the definition of the \prob{Tree-Residue Vertex-Breaking} problem from Section~\ref{hex_grids_subsection}:

\trvbprob*

The problem we will reduce from is one of the variants of \prob{Tree-Residue Vertex-Breaking} from \cite{trvb}:

\begin{theorem} {\rm \cite{trvb}}
The \prob{$6$-Regular Breakable Planar Tree-Residue Vertex-Breaking} problem asks the same question as the \prob{Tree-Residue Vertex-Breaking} problem, but restricts the inputs to be a planar multigraph whose vertices are each breakable and have degree exactly~$6$. The \prob{$6$-Regular Breakable Planar Tree-Residue Vertex-Breaking} problem is NP-complete.
\end{theorem}

The main work of this reduction is constructing a gadget to simulate a degree-$6$ breakable vertex. The desired behavior of the gadget is shown in Figure~\ref{fig:degree_6_gadget_idea}. If we define a wire to be a path of pixels then the idea of the gadget is to connect $6$ incoming wires in such a way that the local constraints on a hypothetical Hamiltonian cycle allow only two possible solutions within the gadget. In one of the locally allowed solutions (Figure~\ref{fig:degree_6_connecting}) the regions inside the six wires are connected through the gadget while in the other (Figure~\ref{fig:degree_6_disconnecting}) the regions are all disconnected at the gadget. Note that the gadget shown in Figure~\ref{fig:degree_6_gadget_idea} is only a schematic and cannot be used as the actual gadget for the reduction because it does not lie on the hexagonal grid.

\begin{figure}[!htbp]
  \centering
  \def\scale{.3}
  \begin{subfigure}[b]{1.6in}
    \centering
    \includegraphics[scale=\scale]{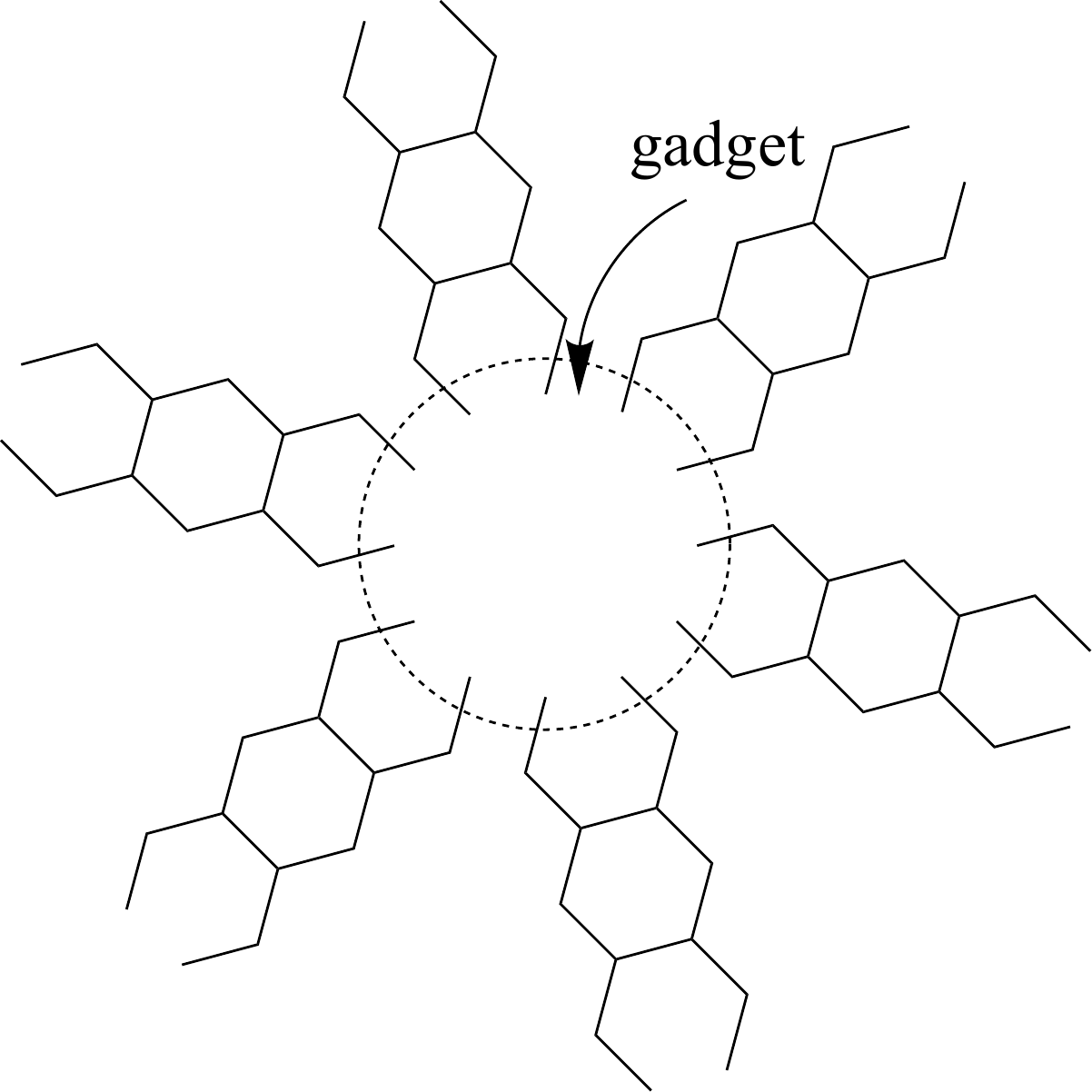}
    \caption{A degree-$6$ breakable gadget connects $6$ wires with only two local solutions}
  \end{subfigure}\hfill
  \begin{subfigure}[b]{1.6in}
    \centering
    \includegraphics[scale=\scale]{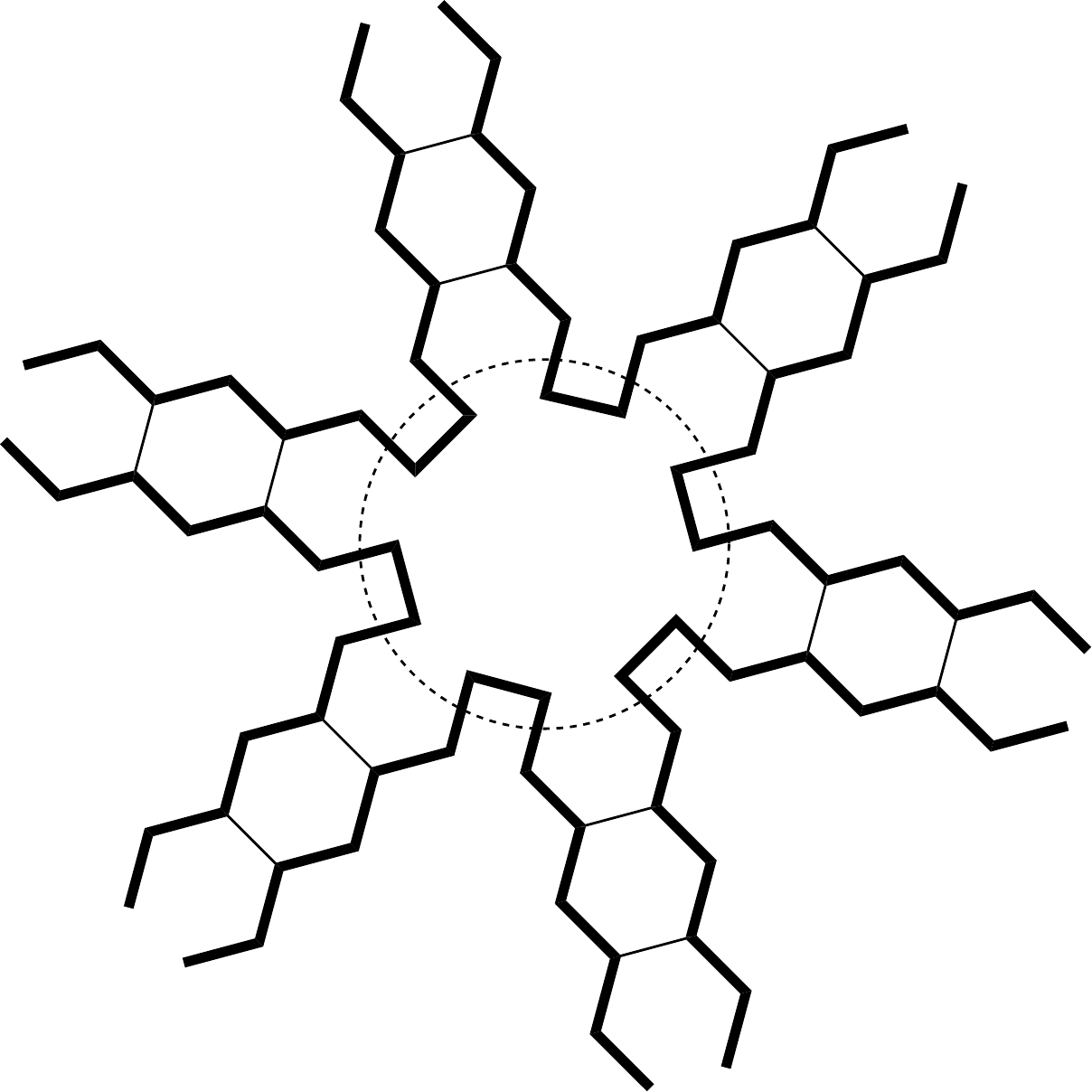}
    \caption{One solution of the gadget connects the regions inside the wires}
    \label{fig:degree_6_connecting}
  \end{subfigure}\hfill
  \begin{subfigure}[b]{1.6in}
    \centering
    \includegraphics[scale=\scale]{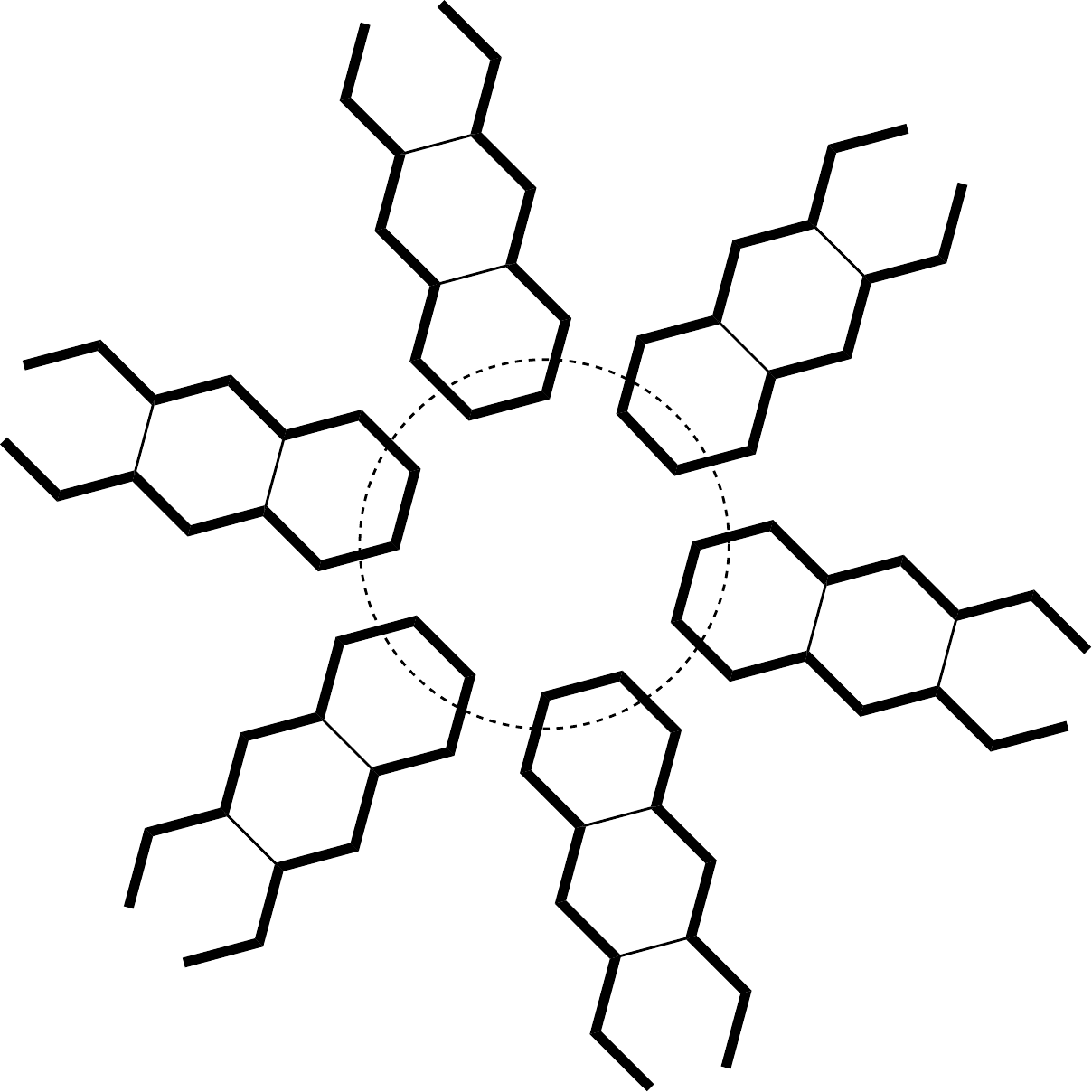}
    \caption{The other solution disconnects the regions inside the wires}
    \label{fig:degree_6_disconnecting}
  \end{subfigure}
  \caption{Desired behavior of a degree-$6$ breakable vertex gadget}
  \label{fig:degree_6_gadget_idea}
\end{figure}

Below, we will (1) provide and analyze a reduction from the \prob{$6$-Regular Breakable Planar Tree-Residue Vertex-Breaking} problem under the assumption of the existance of a degree-$6$ breakable vertex gadget and (2) provide the gadget and prove that it has the desired behavior.

\subsection{Reduction}

Suppose we have a degree-6 breakable vertex gadget with the desired behavior. Then to complete the reduction from the \prob{$6$-Regular Breakable Planar Tree-Residue Vertex-Breaking} problem, we simply lay out the given multigraph in the plane, replace every vertex with the degree-$6$ breakable vertex gadget, and extend a wire for each edge from the gadget of one endpoint to the gadget of the other. Since finding a planar embedding for a graph is a polynomial time operation (provided one exists), such a reduction can be completed in polynomial time.

Below, we define more precisely what constraints a degree-$6$ vertex gadget must satisfy and then prove that the above reduction is correct. In short, the idea is that there is a correspondence between breaking a vertex and choosing the gadget solution shown in Figure~\ref{fig:degree_6_disconnecting}; under this correspondence the shape of the region inside the candidate set of edges is the same as the post-breaking multigraph. This region is connected and hole-free if and only if the post-breaking graph is connected and acyclic. Since a region has as its boundary a single cycle if and only if the region is connected and hole-free, we conclude that the candidate set of edges is a Hamiltonian cycle if and only if the answer to the corresponding \prob{$6$-Regular Breakable Planar Tree-Residue Vertex-Breaking} instance is ``yes.''

Suppose we start with an instance of \prob{$6$-Regular Breakable Planar Tree-Residue Vertex-Breaking} consisting of multigraph $M$ and produce via the above reduction a grid graph $G$.

\begin{definition}
Define a \emph{candidate solution} for the \prob{Hexagonal Thin Grid Graph Hamiltonicity} instance $G$ to be a set of edges $C$ in $G$ satisfying the following constraints: (1) every vertex in $G$ is the endpoint of exactly two edges in $C$ and (2) no one gadget or wire contains a cycle of edges from $C$ entirely inside it.
\end{definition}

Note that failing to satisfy either of the constraints in the above definition is sufficient to disqualify a set of edges from being a Hamiltonian cycle. Thus the question of whether a Hamiltonian cycle exists in $G$ is the same as the question of whether a candidate solution consisting of just one cycle exists. 

Before, we described a degree-$6$ breakable vertex gadget by saying that the local constraints on a hypothetical Hamiltonian cycle allow only two possible solutions within the gadget. To make this more precise, intersecting the set of edges in the gadget with the set of edges in a candidate solution should only have two possible results. As before, the two possibilities are shown in Figure~\ref{fig:degree_6_gadget_idea}. We will refer to the posibilities in Figure~\ref{fig:degree_6_connecting} and Figure~\ref{fig:degree_6_disconnecting} as the connecting and disconnecting solutions respectively. We will refer to a vertex gadget using one of these solutions as a connecting or disconnecting vertex gadget.

Consider a wire between two gadgets and let $C$ be any candidate solution. In both possible solutions of a gadget, there are two edges from the gadget entering into the wire. Thus two edges from $C$ must enter the wire from each end. Simply applying the definition of a candidate solution (no cycles within a wire and every vertex must touch two edges), we can conclude that the edges in the wire that belong to $C$ must be the boundary edges of the wire. For example, see Figure~\ref{fig:wire_example}.

\begin{figure}[!htbp]
  \centering
  \def\scale{1}
  \begin{subfigure}[b]{2.5in}
    \centering
    \includegraphics[scale=\scale]{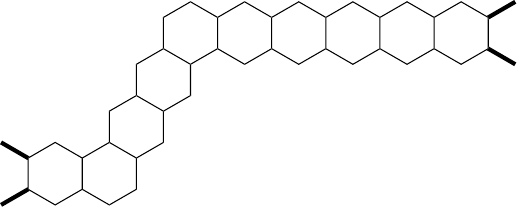}
    \caption{}
  \end{subfigure}\hfill
  \begin{subfigure}[b]{2.5in}
    \centering
    \includegraphics[scale=\scale]{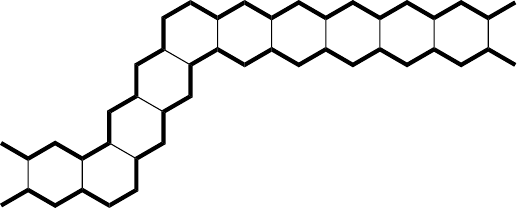}
    \caption{}
  \end{subfigure}
  \caption{The wire shown connects two degree-$6$ breakable vertex gadgets. The bold edges on the left are the edges that must belong to any candidate solution because of the behavior of the degree-$6$ vertex gadgets. The bold edges on the right show the resulting forced behavior in the wire.}
  \label{fig:wire_example}
\end{figure}

As a result, a candidate solution is completely constrained by the choice of behavior at each degree-$6$ vertex gadget. We can identify the disconnecting solution of a vertex gadget with the choice of breaking the corresponding vertex and identify the connecting solution with the choice of leaving the vertex unbroken. Under this correspondence we have a bijection between the choice of candidate solution in the \prob{Hexagonal Thin Grid Graph Hamiltonicity} instance $G$ and the choice of what vertices to break in the \prob{$6$-Regular Breakable Planar Tree-Residue Vertex-Breaking} instance $M$. We show below that the candidate solutions which are Hamiltonian cycles correspond under this bijection with the choices of vertices to break whose breaking converts $M$ into a tree. As a result, a Hamiltonian cycle exists in $G$ if and only if it is possible to break vertices in $M$ so as to obtain a tree, so we conclude that the reduction is correct.

By definition, a candidate solution $C$ is a disjoint cycle cover of $G$ (since every vertex has two edges in $C$ incident on it). Therefore $C$ seperates the plane into several regions. Among these regions, consider the ones which contain the interior of a wire. Let $R$ be the union of these regions. 

Looking at the possible behaviors of a candidate solution in a vertex gadget, it is easy to see that $R$ will consist of the interior of each wire, one connected hole-free region from each connecting vertex gadget, and six connected hole-free regions from each disconnecting vertex gadget. In fact, the boundary of $R$ is exactly the set of edges $C$. 

Consider the above decomposition of $R$ into sub-regions. The choice of candidate solution $C$ corresponds with a set $S$ of vertices to break in $M$. Let $M_S$ be the version of $M$ with the vertices of $S$ broken. The sub-regions of $R$ exactly correspond with the vertices and edges of $M_S$. If $u$ and $v$ are vertices in $M$ then the edge $(u,v)$ in $M$ also occurrs in $M_S$ and corresponds to the interior of the wire between the vertex gadgets for $u$ and $v$. If $u \not\in S$ is a vertex of $M$ then $u$ is also a vertex in $M_S$ and corresponds to the sub-region of $R$ from the vertex gadget for $u$. If $u \in S$ is a vertex of $M$ then there are six vertices in $M_S$ that result from the breaking of $u$, and these six vertices corresponds to the six sub-regions of $R$ from the vertex gadget for $u$.

Clearly, the vertices and edges of $M_S$ correspond to the sub-regions of $R$. Furthermore, it can be easily verified that two sub-regions of $R$ touch if and only if those two sub-regions correspond to a vertex $v$ and an edge $e$ in $M_S$ such that the $v$ is an endpoint of $e$. Furthermore, there is no hole in $R$ at such a point of contact between two sub-regions. We show below using these facts that $M_S$ is connected if and only if $R$ is connected and that $M_S$ is acyclic if and only if $R$ is hole-free.

$M_S$ is connected if and only if there exists a path through $M_S$ from any edge or vertex $a$ to any other edge or vertex $b$. Such a path can be expressed as a list of edges and vertices $a=x_0, x_1, \ldots, x_k = b$ such that every consecutive pair $x_i$ and $x_{i+1}$ consists of one edge and one vertex that is an endpoint of that edge. We can create a corresponding list of sub-regions of $R$ of the form $y_0, y_1, \ldots, y_k$ where $y_i$ is the sub-region corresponding to edge or vertex $x_i$. Note that in this list, every consecutive pair $y_i$ and $y_{i+1}$ consists of two touching subregions. Thus we see that $M_S$ is connected if and only if for every pair of sub-regions $a', b'$, there exists a list of the form $a'=y_0, y_1, \ldots, y_k=b'$ with $y_i$ touching $y_{i+1}$ for each $i$. Since every sub-region is itself connected, this last condition is equivalent to the condition that there exist a path in $R$ between any two points; or in other words, $M_S$ is connected if and only if $R$ is connected.

Consider any cycle in $M_S$ consisting of vertices and edges $x_0, x_1, \ldots, x_k = x_0$. This cycle corresponds to a cycle of sub-regions $y_0, y_1, \ldots, y_k = y_0$. Such a cycle of sub-regions will have an inner boundary, or in other words a hole. On the other hand, for any hole in $R$, we can list the subregions going around that hole: $y_0, y_1, \ldots, y_k = y_0$. Since each sub-region is hole free and each sub-region contact point does not have a hole, this list of subregions will correspond to a cycle of vertices and edges $x_0, x_1, \ldots, x_k = x_0$ in $M_S$. Thus we see that $M_S$ is acyclic if and only if $R$ is hole-free.

$M_S$ is a tree, or in other words connected and acyclic, if and only if $R$ is hole-free and connected. But $R$ is hole-free and connected if and only if the boundary of $R$ (which happens to be $C$) is exactly one cycle. Since $C$ is a cycle cover of $G$, $C$ is a Hamiltonian cycle if and only if $C$ consists of exactly one cycle. Putting that all together, we see that $M_S$ is a tree if and only if $C$ is a Hamiltonian cycle in $G$.

As you can see, the constructed grid graph is Hamiltonian if and only if the input \prob{$6$-Regular Breakable Planar Tree-Residue Vertex-Breaking} instance is a ``yes'' instance, and therefore the reduction is correct. 

\subsection{Degree-$6$ breakable vertex gadget}

Thus, all that is left is to demonstrate a degree-$6$ breakable vertex gadget of the desired form. The gadget we will use is shown in Figure~\ref{fig:degree_6_gadget}.

\begin{figure}[!htbp]
  \centering
  \begin{subfigure}[b]{2.5in}
    \centering
    \includegraphics[scale=0.6]{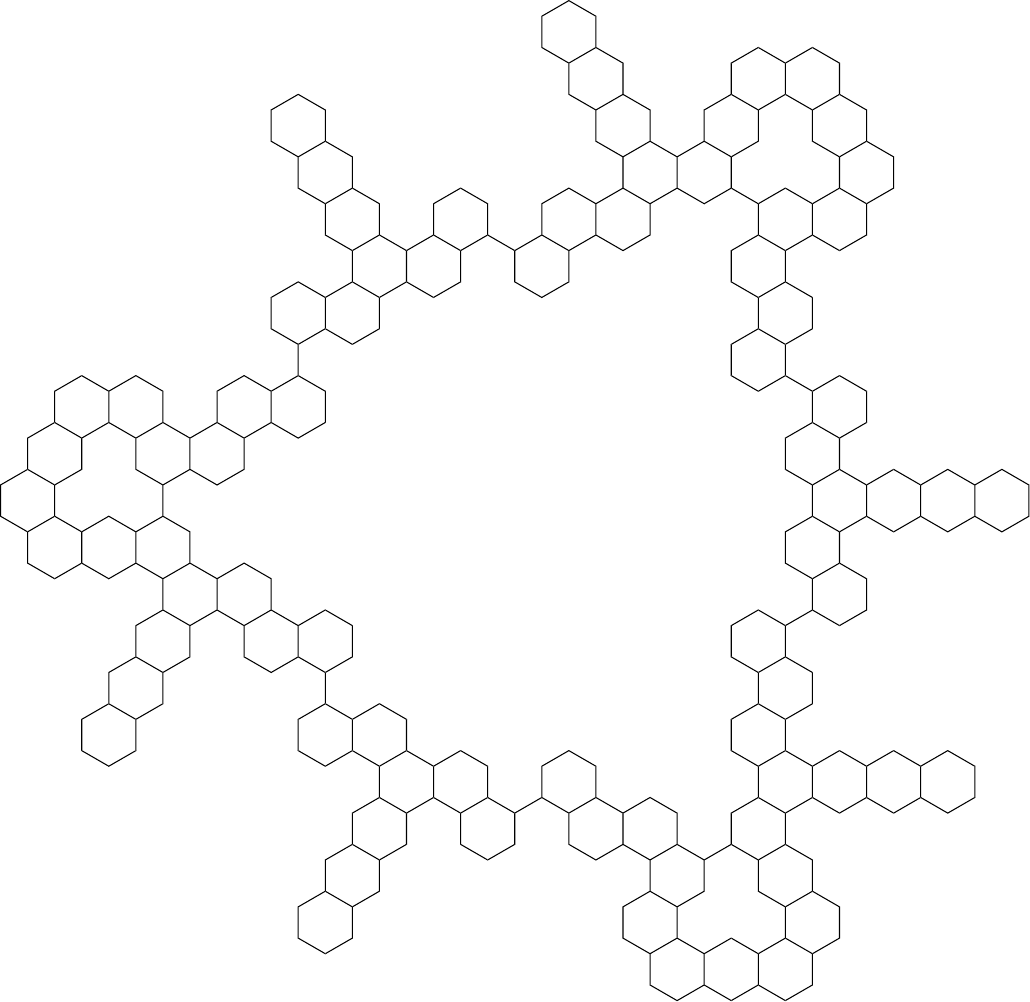}
    \caption{A degree-$6$ breakable vertex gadget (with six wires).}
    \label{fig:degree_6_gadget}
  \end{subfigure}\hfill
  \begin{subfigure}[b]{2.5in}
    \centering
    \includegraphics[scale=0.6]{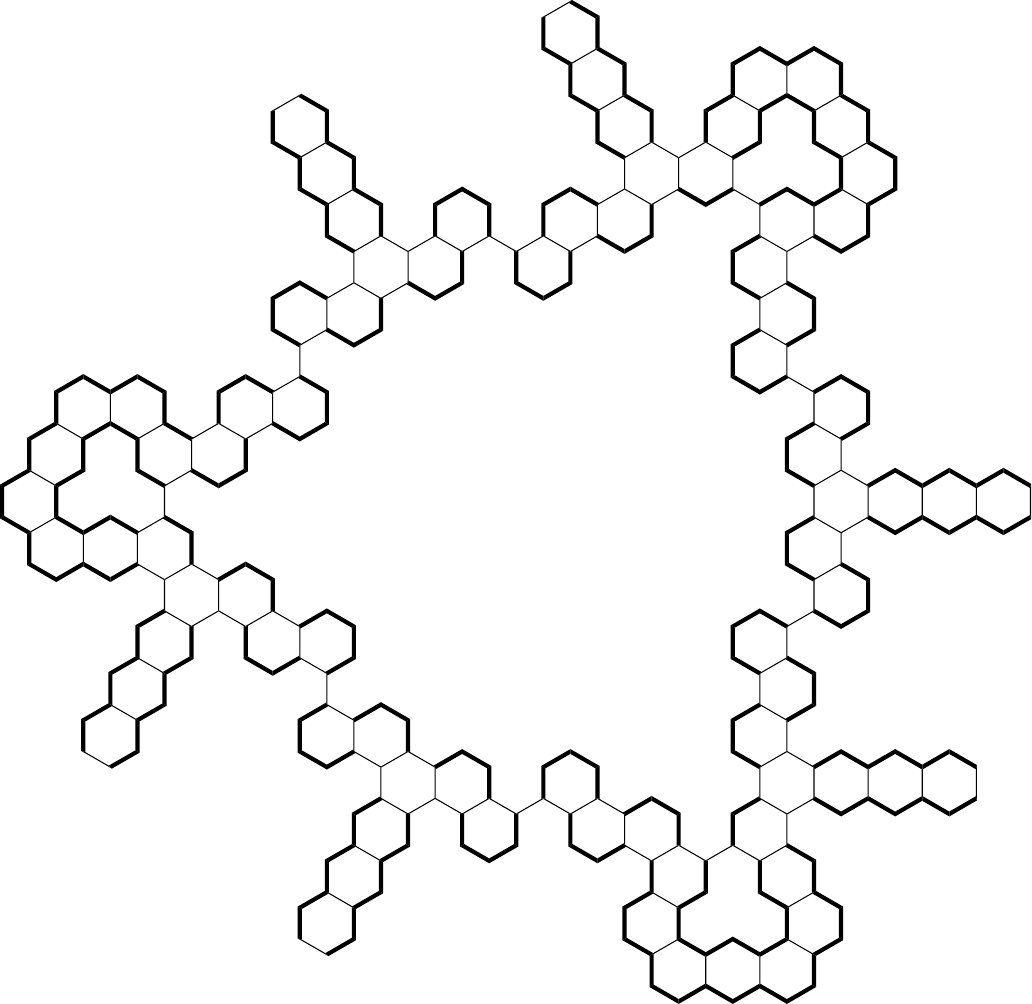}
    \caption{The gadget with edges that must be in a Hamiltonian cycle in bold.}
    \label{fig:degree_6_gadget_solution_forced}
  \end{subfigure}\hfill
  \begin{subfigure}[b]{2.5in}
    \centering
    \includegraphics[scale=0.6]{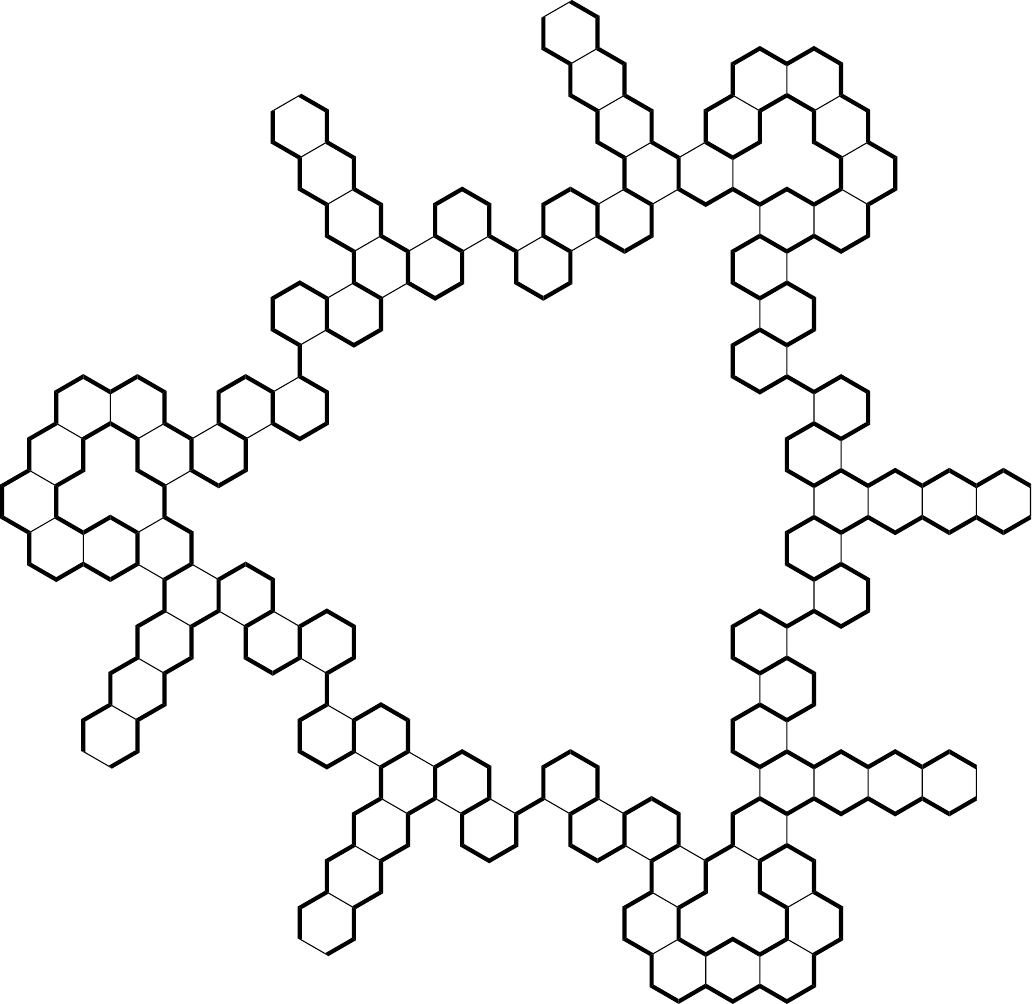}
    \caption{The first solution of the gadget with edges chosen for the Hamiltonian cycle bold. Note that the regions inside the wires are connected via this gadget.}
    \label{fig:degree_6_gadget_solution_1}
  \end{subfigure}\hfill
  \begin{subfigure}[b]{2.5in}
    \centering
    \includegraphics[scale=0.6]{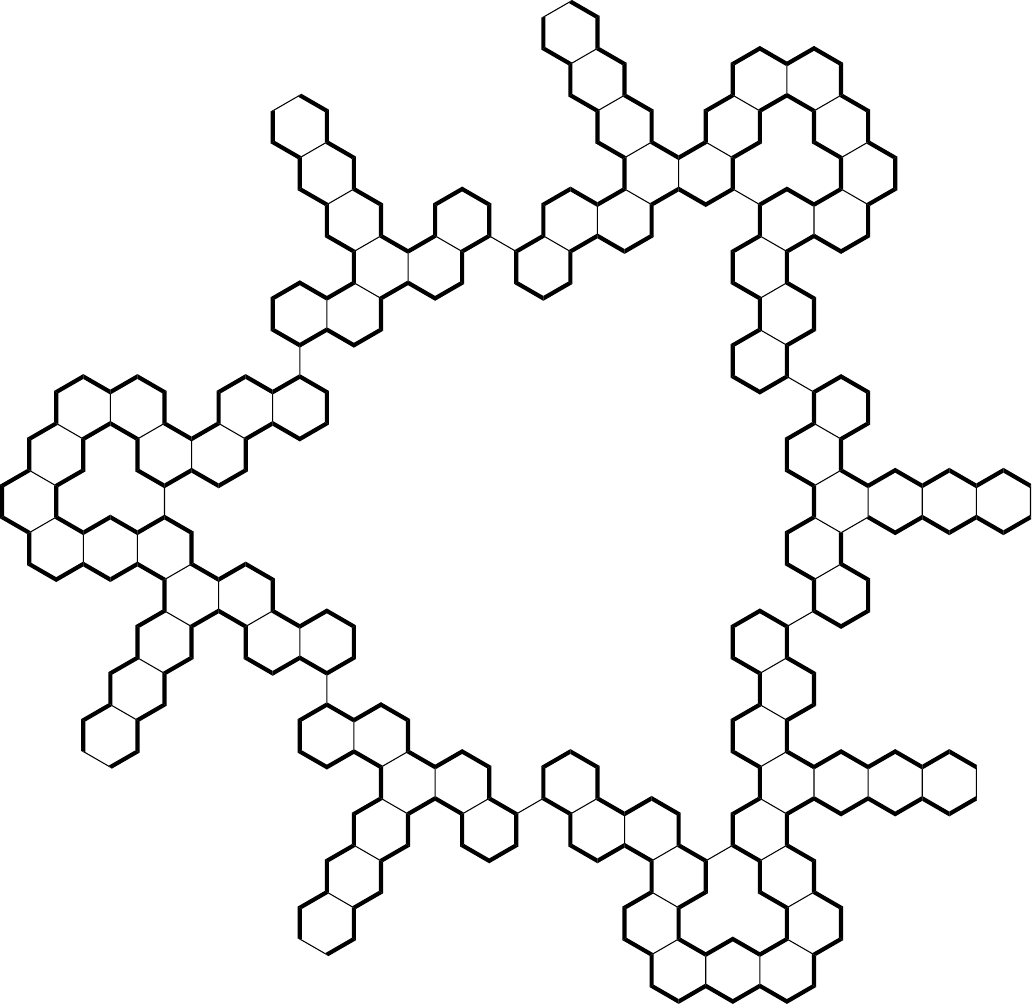}
    \caption{The second solution of the gadget with edges chosen for the Hamiltonian cycle bold. Note that the regions inside the wires are disconnected via this gadget.}
    \label{fig:degree_6_gadget_solution_2}
  \end{subfigure}
  \caption{A degree-$6$ breakable vertex gadget together with the two possible solutions.}
  \label{fig:degree_6_gadget_all}
\end{figure}

We show below that the gadget only has the two desired solutions.

\begin{theorem}
The gadget shown in Figure~\ref{fig:degree_6_gadget} has exactly two possible solutions and they correspond with the solutions shown in Figure~\ref{fig:degree_6_gadget_idea}.
\end{theorem}

\begin{proof}

Figure~\ref{fig:degree_6_gadget_solution_forced} shows the gadget from Figure~\ref{fig:degree_6_gadget}, but with all edges that must be in any candidate solution bold. These constraints can be derived purely from the fact that every vertex must be the endpoint of exactly two edges in a candidate solution.

The gadget contains six contiguous regions of pixels, arranged in a cycle with a single edge that is not part of any pixel between adjacent regions. We claim that either all of these single edges must be in the candidate solution or none of them. It can be verified that these two possibilities lead, again via the constraints on candidate solutions, to the two solutions shown in Figures~\ref{fig:degree_6_gadget_solution_1} and~\ref{fig:degree_6_gadget_solution_2}.

All that's left to show is that our claim is correct: that using only some of the single edges in a candidate solution is impossible. Suppose for the sake of contradiction that a candidate solution exists which uses some but not all of these edges. Going around the cycle of regions, there will be some region where we transition from using the single edge to not using it. That region will have exactly one of its two single edges in the candidate solution. Thus, the candidate solution will enter or exit the region exactly three times: twice through the wire that exists the gadget from this region and once through the single edge. Since this is impossible for a cycle cover, we arrive at the desired contradiction.
\end{proof}

\FloatBarrier
\section{Square Polygonal Grid Graph Hamiltonicity is NP-complete}
\FloatBarrier
\label{sec:square polygonal}

In this section, we show that the \prob{Square Polygonal Grid Graph Hamiltonicity} problem is NP-complete. Membership in NP is trivial, while NP-hardness follows from a polynomial-time reduction. We reduce from the \prob{Planar Monotone Rectilinear 3SAT} problem---which was shown NP-hard in \cite{deBerg}---to the \prob{Square Polygonal Grid Graph Hamiltonicity} problem. In order to state the \prob{Planar Monotone Rectilinear 3SAT} problem, we first need some preliminary terms:

\begin{definition}
A \defn{rectilinear embedding} of a CNF formula into a plane is a planar embedding of the variable-clause bipartite graph for the CNF formula into the plane such that the vertices (variables and clauses) are mapped to horizontal segments, the edges are mapped to vertical segments, and the variables all lie on the $x$ axis. Furthermore, a \defn{monotone rectilinear embedding} of a CNF formula into a plane is a rectilinear embedding with the additional property that the clauses positioned above the $x$ axis contain only positive literals while clauses positioned below the $x$ axis contain only negative literals. 
\end{definition}

Figure~\ref{fig:rectilinear_3sat} shows an example of a rectilinear embedding. This embedding could represent several possible CNF formulas, such as $(x_1 \vee \overline{x_2} \vee \overline{x_3}) \wedge (x_3 \vee \overline{x_4} \vee x_5) \wedge (\overline{x_2} \vee x_3 \vee \overline{x_5})$ or $(\overline{x_1} \vee \overline{x_2} \vee \overline{x_3}) \wedge (\overline{x_3} \vee x_4 \vee x_5) \wedge (x_2 \vee x_3 \vee x_5)$. If, however, this is a monotone rectilinear embedding then a specific CNF formula, $(x_1 \vee x_2 \vee x_3) \wedge (x_3 \vee x_4 \vee x_5) \wedge (\overline{x_2} \vee \overline{x_3} \vee \overline{x_5})$ is implied.

\begin{figure}[!htbp]
    \centering
    \includegraphics[scale = 0.75]{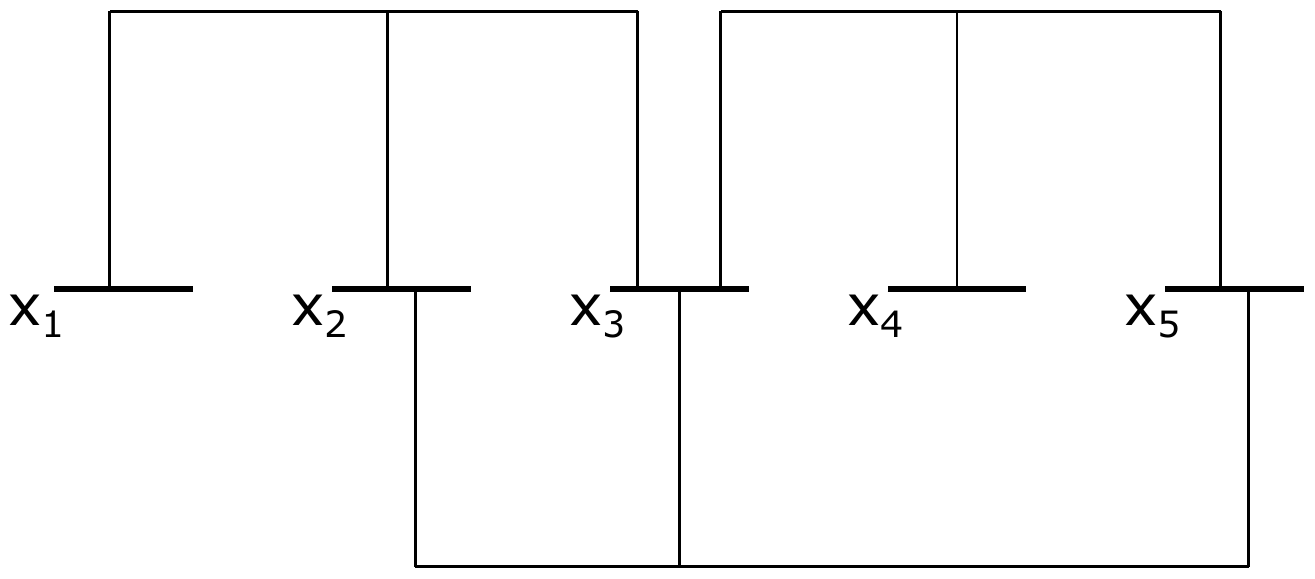}
    \caption{A rectilinear embedding}
    \label{fig:rectilinear_3sat}
\end{figure}

Now the problem we are reducing from can be stated as follows:

\begin{problem}
The problem \prob{Planar Monotone Rectilinear 3SAT} decides, for every monotone rectilinear embedding of a 3-CNF formula into the plane, whether the given instance is satisfiable.
\end{problem}

Our overall strategy for reducing \prob{Planar Monotone Rectilinear 3SAT} to \prob{Hamiltonicity of Square Polygonal Grid Graphs} is the following:

\begin{enumerate}
\item We begin by describing several simple gadgets.
\begin{enumerate}
\item A wire gadget consists of a path of pixels.
\item A one-enforcer gadget, when inserted into a wire, enforces the fact that any Hamiltonian cycle in the grid graph must pass through the wire only once, zig-zagging through the wire in one direction.
\item A two-enforcer gadget, when inserted into a wire, enforces the fact that any Hamiltonian cycle in the grid graph must pass through the wire twice, once in each direction.
\end{enumerate}
\item We can then combine these gadgets to form variable and clause gadgets, which we will use to simulate a \prob{Planar Monotone Rectilinear 3SAT} instance.
\begin{enumerate}
\item A variable gadget consists of a loop of wire, which when correctly inserted into a one-enforced wire enforces certain constraints on any Hamiltonian cycle (provided one exists). In particular, among the two wires in the gadget (the top and bottom halves of the loop of wire), one will be one-enforced, and the other will be two-enforced. This choice corresponds to a choice of value for the variable.
\item A clause gadget attaches to three variable gadgets (attaching to either the top or bottom wires of each) in such a way that the previously stated facts about variable gadgets continue to hold. If all three attaching points connect the clause gadget to one-enforced wire (due to the choices at the variable gadgets) then there exists no Hamiltonian cycle in the grid graph. Conversely, whenever this fails to occur, it is possible to modify one of the pieces-of-cycle passing through the variable gadgets to also pass through every point in the clause gadget.
\end{enumerate}
\item Next, we assemble a grid graph out of these gadgets to simulate a \prob{Planar Monotone Rectilinear 3SAT} instance. We create a variable gadget for every variable and connect them in a loop with one-enforcers appropriately placed. We then add clause gadgets outside the loop and inside the loop to represent the positive and negative clauses. For an example, see Figure~\ref{fig:square_polygonal_example}.
\item If there exists a satisfying assignment then we can construct a cycle passing through the one-enforced loop and through the variable gadgets in such a way that each clause gadget attaches to at least one two-enforced wire (among the three variable gadget wires that it attaches to). Then by the clause gadget properties listed above, we can modify that cycle to also pass through every point in each clause gadget. This yields a Hamiltonian cycle. If there exists no satisfying assignment then any cycle passing through the one-enforced wire loop and variable gadgets will result in at least one clause gadget being attached to three one-enforced wires. Then by the clause gadget properties listed above, no Hamiltonian cycle exists. This shows that the reduction is answer preserving. 
\end{enumerate}

\begin{figure}[!htbp]
  \centering
  \subcaptionbox{The \prob{Planar Monotone Rectilinear 3SAT} instance $(x_1 \vee x_2 \vee x_3) \wedge (x_1 \vee x_3 \vee x_4) \wedge (\overline{x_1} \vee \overline{x_2} \vee \overline{x_4}) \wedge (\overline{x_2} \vee \overline{x_3} \vee \overline{x_3}) \wedge (\overline{x_2} \vee \overline{x_3} \vee \overline{x_4})$}
    {\includegraphics[width=.5\textwidth]{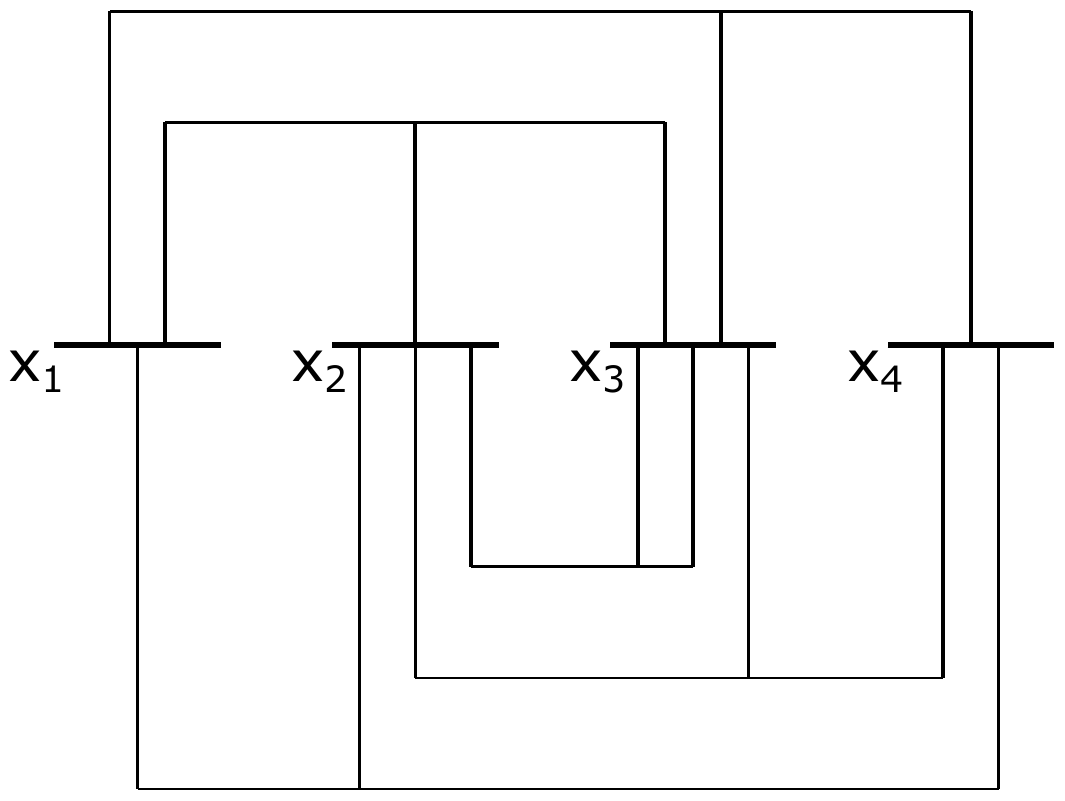}}\hfill
  \subcaptionbox{The corresponding \prob{Hamiltonicity of Square Polygonal Grid Graphs} instance}
    {\includegraphics[width=.5\textwidth]{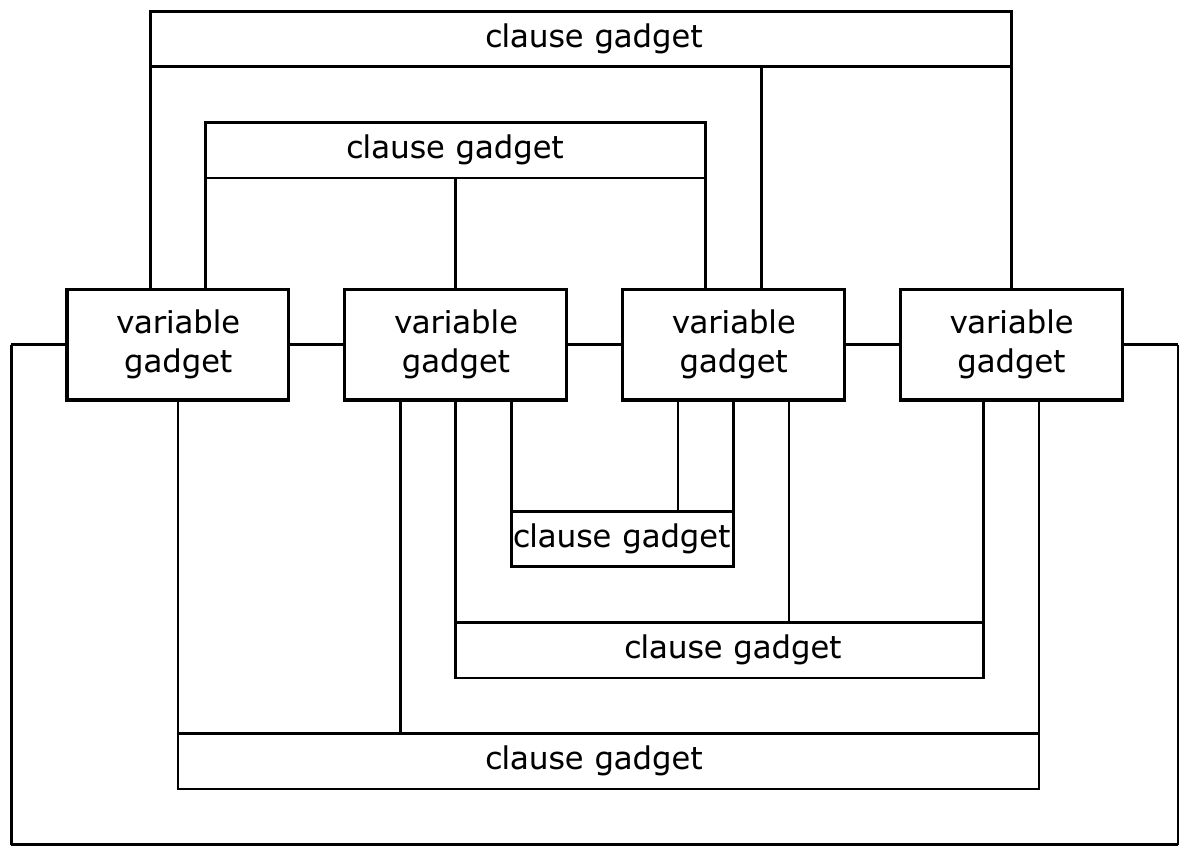}}
  \caption{An example of the reduction from \prob{Planar Monotone Rectilinear 3SAT} to \prob{Hamiltonicity of Square Polygonal Grid Graphs}}
  \label{fig:square_polygonal_example}
\end{figure}

We follow this outline below to prove that

\begin{theorem}
There exists a polynomial time reduction from the \prob{Planar Monotone Rectilinear 3SAT} problem to the \prob{Hamiltonicity of Square Polygonal Grid Graphs} problem.
\end{theorem}
and therefore that 

\begin{corollary}
The \prob{Hamiltonicity of Square Polygonal Grid Graphs} problem is NP-complete. 
\end{corollary}

\FloatBarrier
\subsection{Simple Gadgets}
\FloatBarrier

\FloatBarrier
\subsubsection{Wires}
\FloatBarrier

The simplest gadget we will use will be the wire, which is simply a path of pixels. As shown in Figure~\ref{fig:wire_gadget}, there are multiple ways to pass a cycle through a wire. The important distinction to make is between ``two-enforced'' wires for which the cycle passes through most pixels twice, once in each direction (the first two solutions in Figure~\ref{fig:wire_gadget}), and ``one-enforced'' wires for which the cycle passes through most pixels once (the final two solutions in Figure~\ref{fig:wire_gadget}). 

\begin{figure}[!htbp]
  \centering
  \subcaptionbox{Gadget}
    {\includegraphics[scale=0.5]{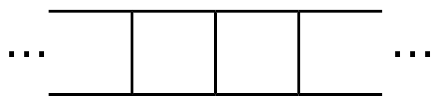}}\hfill
  \subcaptionbox{Solutions}
    {\includegraphics[scale=0.5]{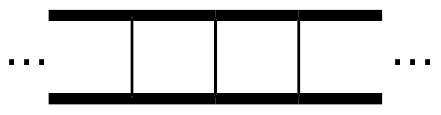}
    \includegraphics[scale=0.5]{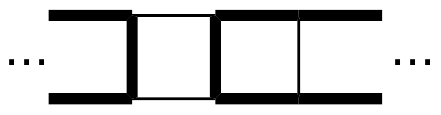}
    \includegraphics[scale=0.5]{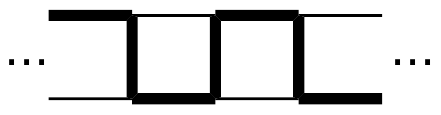}
    \includegraphics[scale=0.5]{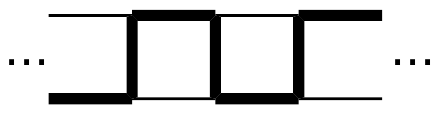}}
  \caption{A wire}
  \label{fig:wire_gadget}
\end{figure}

Note that the distinction between one-enforced and two-enforced wires must propagate through a wire by a simple contradiction argument: if a cycle goes into a wire once from one side and comes out twice from the other then it enters that wire an odd number of times in all, which is impossible for a cycle.

Furthermore, as seen in Figure~\ref{fig:wire_turn}, even if the wire turns (provided the turns are not too close), it is possible to propagate this property. 

\begin{figure}[!htbp]
  \centering
  \hspace{.5in}
  \subcaptionbox{Gadget}
    {\includegraphics[scale=0.5]{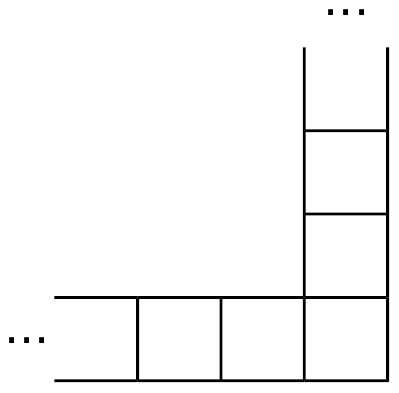}}\hfill
  \subcaptionbox{Solutions}
    {\includegraphics[scale=0.5]{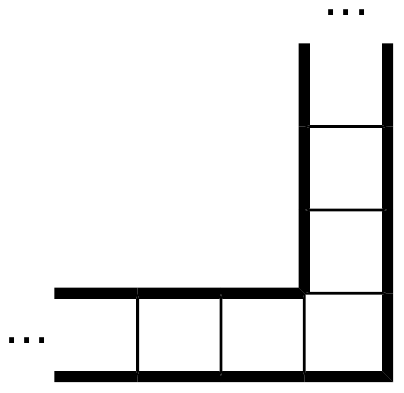}
    \includegraphics[scale=0.5]{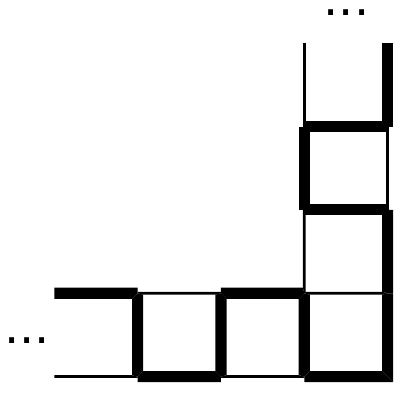}
    \includegraphics[scale=0.5]{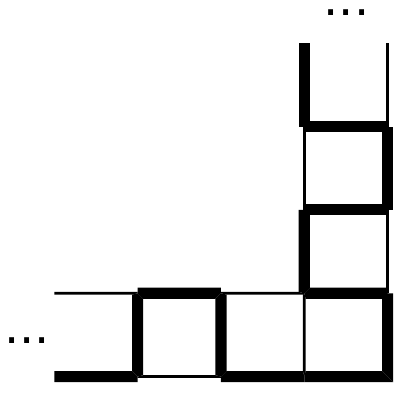}}
  \hspace{.5in}
  \caption{A turning wire}
  \label{fig:wire_turn}
\end{figure}

\FloatBarrier
\subsubsection{One-Enforcers}
\FloatBarrier

A one-enforcer, when inserted into a wire causes the wire to be one-enforced in a particular parity. The gadget is shown in Figure~\ref{fig:one_enforcer_gadget}. The two possible solutions are shown in Figure~\ref{fig:one_enforcer_satisfied}.

\begin{figure}[!htbp]
  \centering
  \subcaptionbox{Gadget
    \label{fig:one_enforcer_gadget}}
    {\includegraphics[scale=0.5]{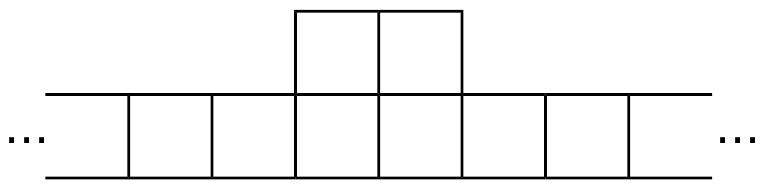}}\hfill
  \subcaptionbox{Solutions
    \label{fig:one_enforcer_satisfied}}
    {\includegraphics[scale=0.5]{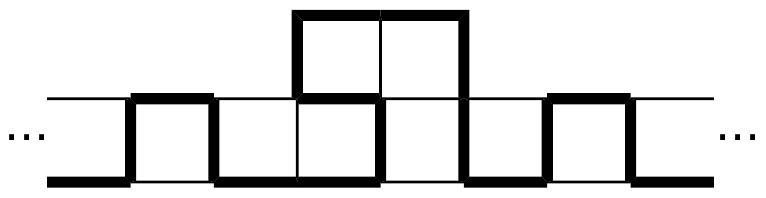}
    \includegraphics[scale=0.5]{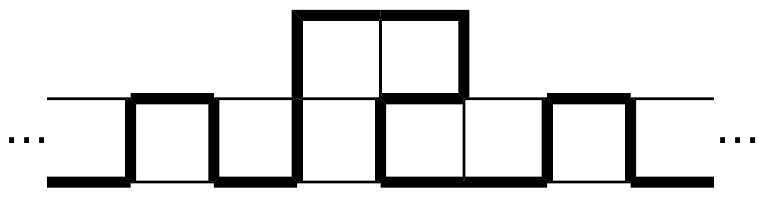}}
  \caption{A one-enforcer gadget}
\end{figure}

In both solutions, the behavior of the cycle in the wire attached to the gadget is the same: a one-enforced wire of a particular parity. The parity is different on each side; in other words, the gadget causes a parity shift in addition to its other effects.

\FloatBarrier
\subsubsection{Two-Enforcers}
\FloatBarrier

Inserting a two-enforcer into a wire causes that wire to be two-enforced. The shape of a two enforcer is shown in Figure~\ref{fig:two_enforcer_gadget}. Every edge adjacent to a vertex with only two neighbors must be in a Hamiltonian cycle if one exists. Thus all the bold edges in Figure~\ref{fig:two_enforcer_satisfied} must be in the cycle. But then the section of gadget in the middle is a two-enforced wire and this property propagates along the wire in both directions away from the gadget. 

\begin{figure}[!htbp]
  \centering
  \hspace{1in}
  \subcaptionbox{Gadget
    \label{fig:two_enforcer_gadget}}
    {\includegraphics[scale=0.5]{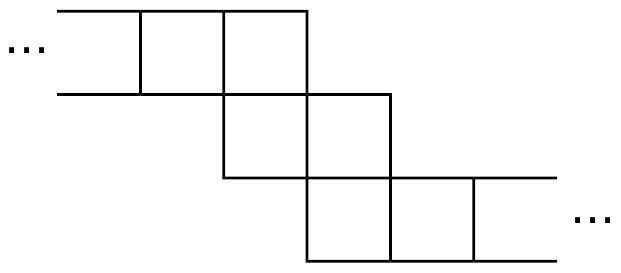}}\hfill
  \subcaptionbox{Solutions
    \label{fig:two_enforcer_satisfied}}
    {\includegraphics[scale=0.5]{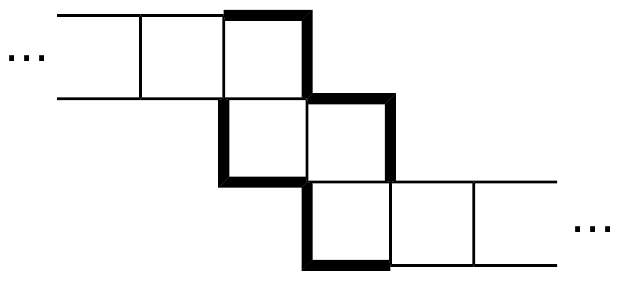}}
  \hspace{1in}
  \caption{A two-enforcer gadget}
\end{figure}

\FloatBarrier
\subsection{Variable Gadget}
\FloatBarrier

A variable gadget consists of a loop of wire of odd width and height 6 that is inserted into a one-enforced wire of a particular parity. Figure~\ref{fig:variable_gadget} shows a variable gadget together with the appropriate one-enforcers necessary for the gadget to properly function.

\begin{figure}[!htbp]
  \centering
  \subcaptionbox{\label{fig:variable_gadget}}
  {\includegraphics[scale=0.5]{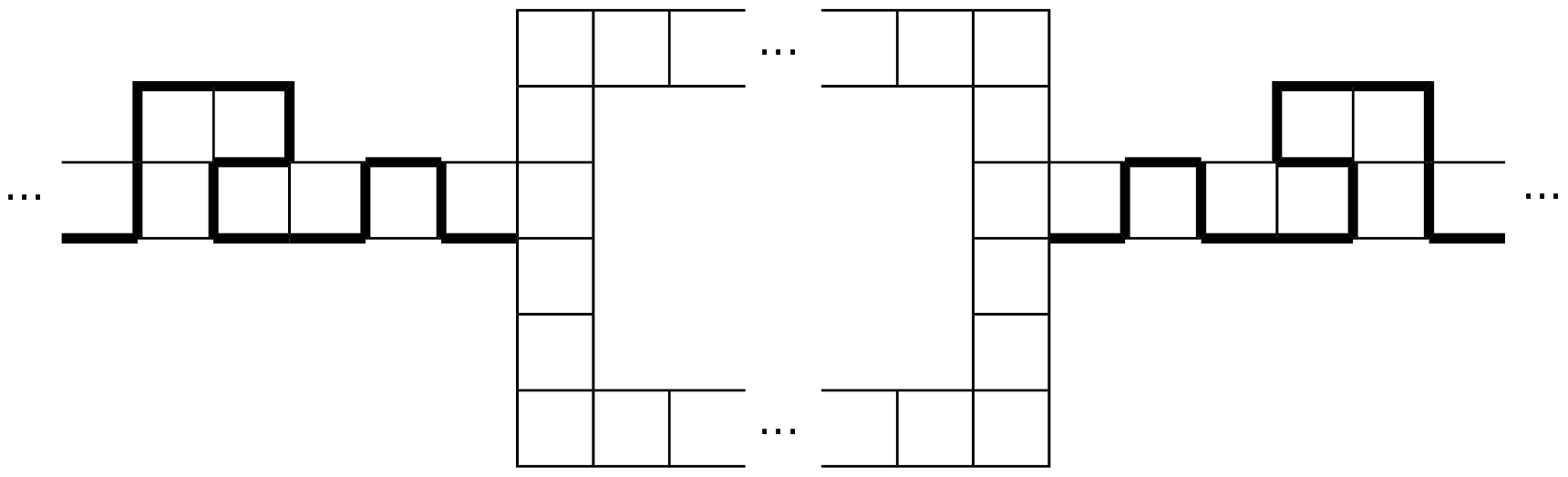}}
  \subcaptionbox{\label{fig:variable_satisfied_1}}
  {\includegraphics[scale=0.5]{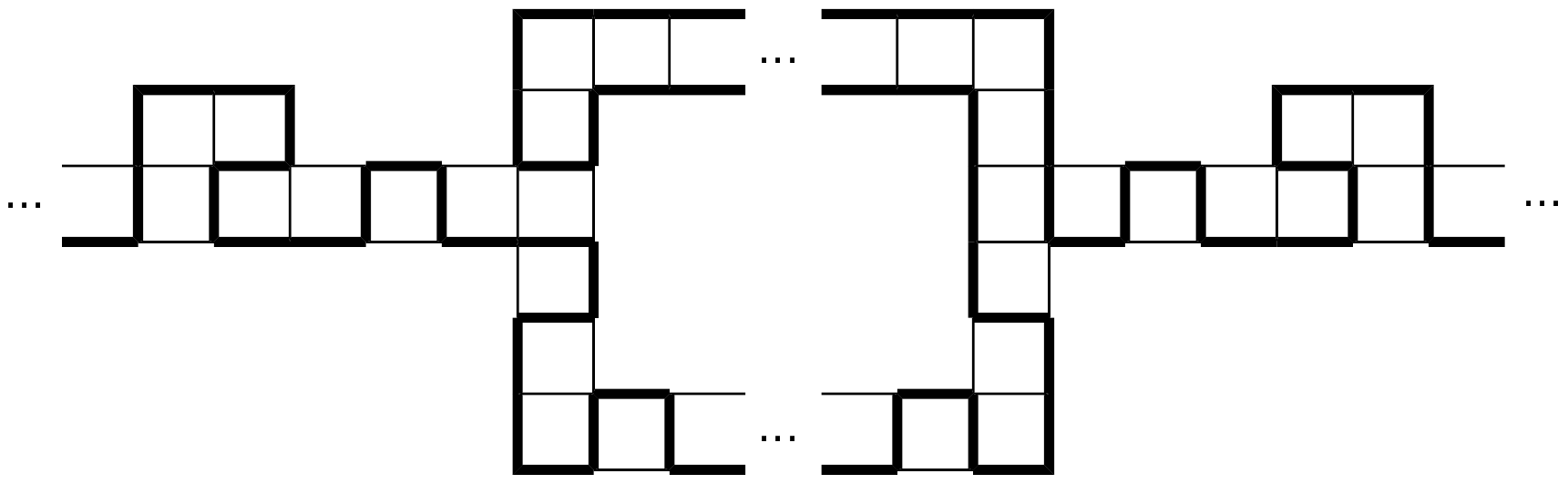}}
  \subcaptionbox{\label{fig:variable_satisfied_2}}
  {\includegraphics[scale=0.5]{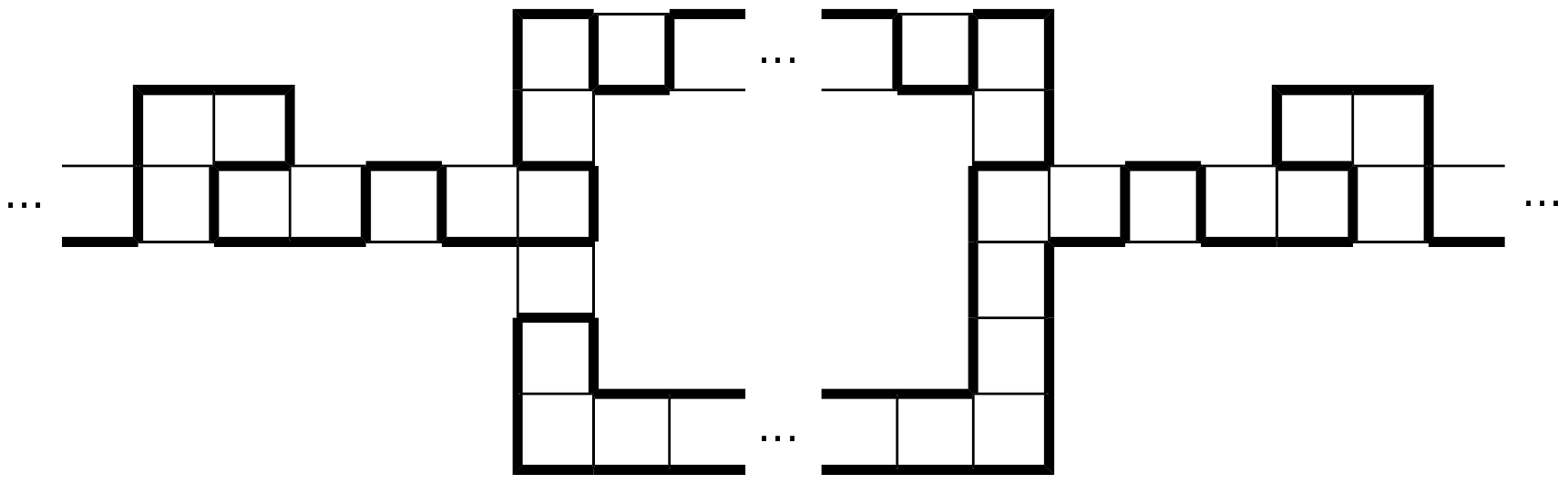}}
  \caption{A variable gadget (Figure~\ref{fig:variable_gadget}) and two ways a cycle can pass through the variable gadget (Figures~\ref{fig:variable_satisfied_1} and \ref{fig:variable_satisfied_2})}
  \label{fig:variable_gadget_all}
\end{figure}

Consider just one half of the gadget (Figure~\ref{fig:variable_half}). The possible ways for a cycle to pass through this gadget half are listed in Figure~\ref{fig:variable_half_satisfied}. The important thing to note is that in all cases one of the two branches is one-enforced and the other is two-enforced.

\begin{figure}[!htbp]
  \centering
  \subcaptionbox{Gadget
    \label{fig:variable_half}}
    {\includegraphics[scale=0.5]{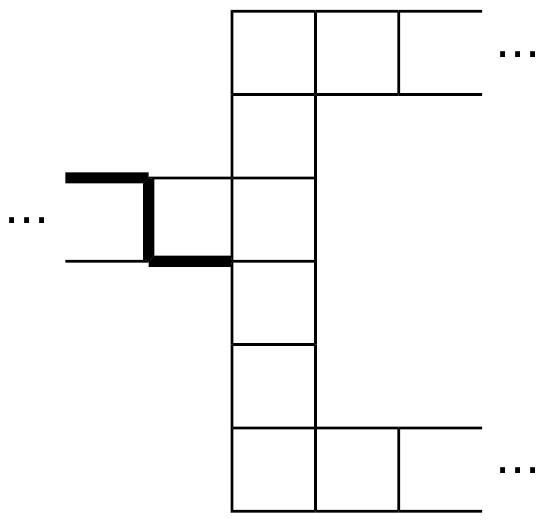}}\hfill
  \subcaptionbox{Solutions
    \label{fig:variable_half_satisfied}}
    {\includegraphics[scale=0.5]{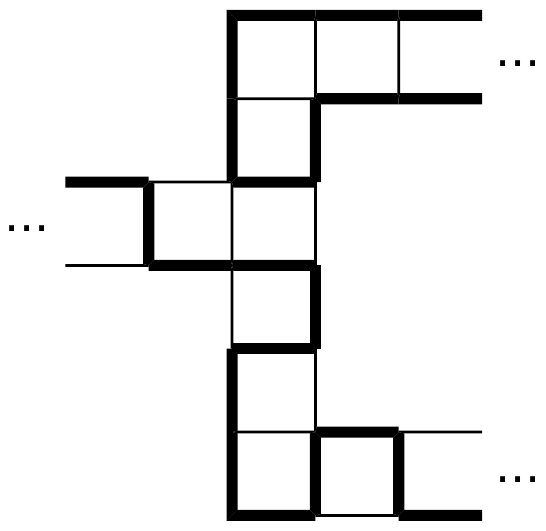}
    \hspace{-.2in}
    \includegraphics[scale=0.5]{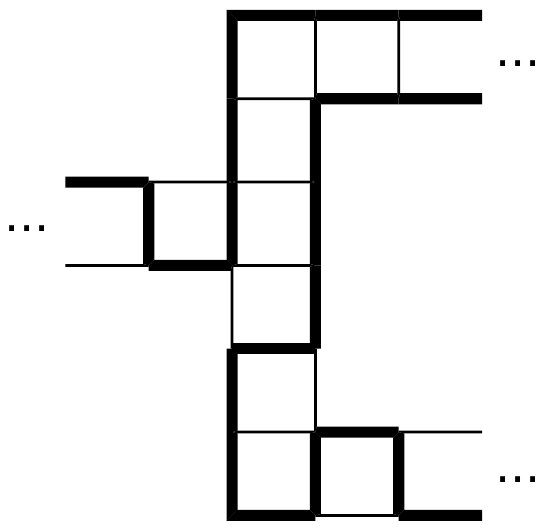}
    \hspace{-.2in}
    \includegraphics[scale=0.5]{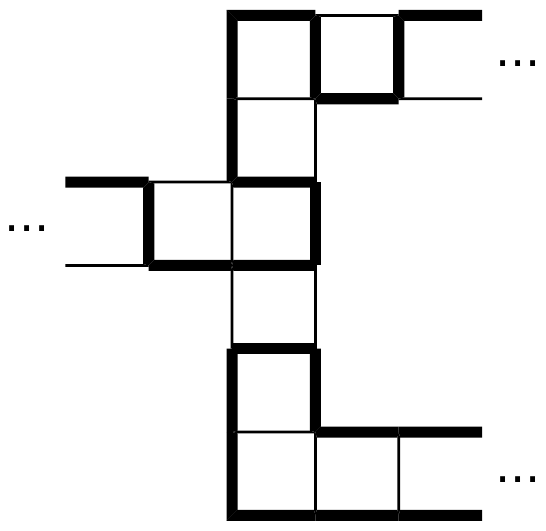}
    \hspace{-.2in}
    \includegraphics[scale=0.5]{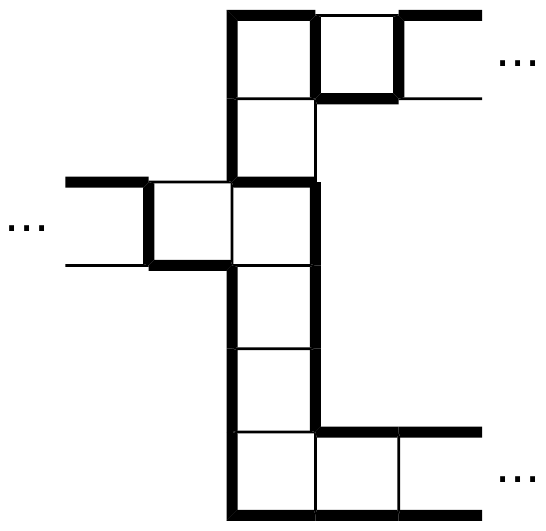}}
  \caption{One half of a variable gadget}
\end{figure}

Putting the variable gadget together from two halves, we can figure out the legal solutions by assigning solutions to the two halves in ways that don't result in isolated cycles. The only two possible ways (up to reflection) for a cycle to pass through a variable gadget are shown in Figures~\ref{fig:variable_satisfied_1} and \ref{fig:variable_satisfied_2}. As you can see, the variable gadget allows a choice: either the top or bottom wire must be two-enforced while the other must be one-enforced.

\FloatBarrier
\subsection{Clause Gadget}
\FloatBarrier

A clause gadget is a long horizontal wire with three two-enforced vertical wires brought vertically down (or up) from the horizontal wire. It is shown in Figure~\ref{fig:clause_gadget} with the two-enforcers omitted but with bold edges as appropriate to indicate that the wires are two-enforced. 

\begin{figure}[!htbp]
    \centering
    \includegraphics[scale=0.5]{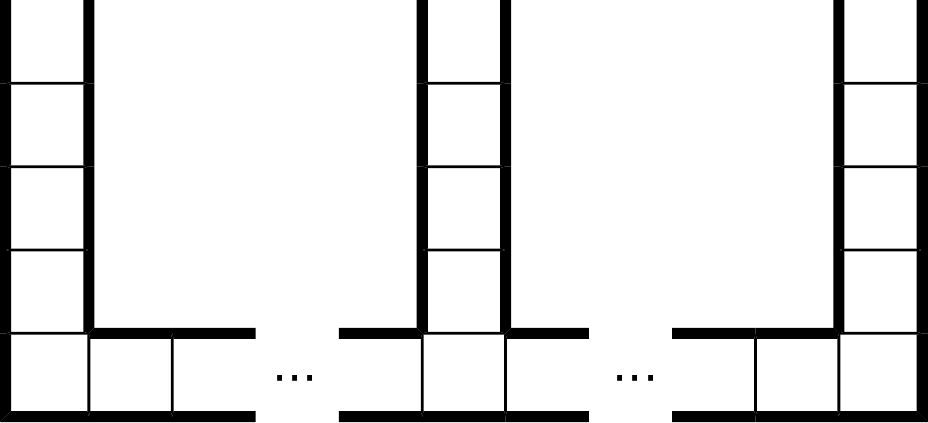} \hspace{1in}
    \includegraphics[scale=0.5]{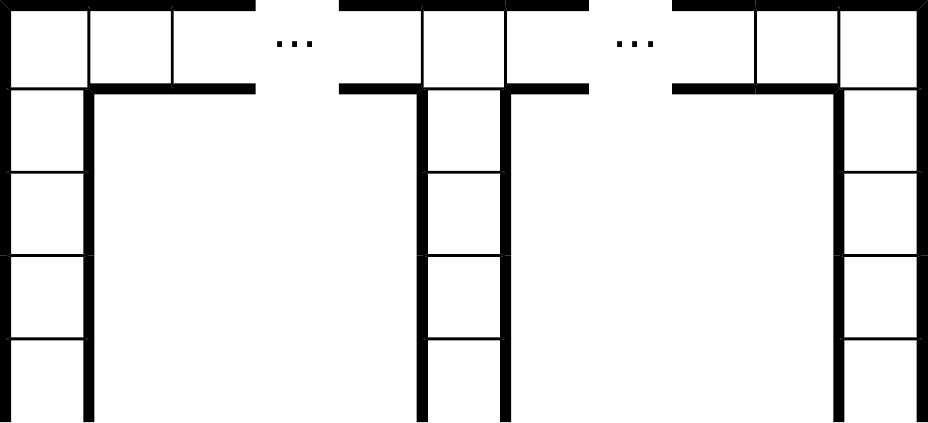}
    \caption{Two clause gadgets}
    \label{fig:clause_gadget}
\end{figure}

When used, this gadget will attach each of the three two-enforced vertical wires to a variable gadget. A variable gadget has a top horizontal wire and a bottom horizontal wire. The vertical two-enforced wire from the clause is connected directly (in a particular place described later) to one of these two wires (the bottom one if the clause's vertical wires go up from the horizontal clause section and the top if the clause's vertical wires go down from the horizontal clause section). This is a modification to the variable gadget, and so we must verify that the variable gadget continues to work as expected (and no other way) with this addition. 

In fact, a rigorous derivation of the behavior of a cycle in a variable gadget relies on a repeated application of the same parity argument. For example, to show that the two wires in a variable half cannot be both one-enforced or both two-enforced, we would argue that the total number of arcs of cycle at the edge of that gadget must be even, and therefore cannot be 3 or 5 (which would occur if the wires were both one- or two-enforced). Similarly, to show that the two-enforced wire from one half of the variable gadget cannot coincide with the one-enforced wire from the other, we argue based on the parity of the number of cycle arcs entering the wire. Since our modification only attaches a two-enforced wire, and since that addition only ever adds an even number of cycle arcs to a region, we see that the same arguments must continue to hold and therefore that the set of possible ways for a cycle to pass through a vertex gadget remains the same (one variable wire will be two-enforced while the other will be one-enforced).

So consider the horizontal variable wire that the clause gadget connects to. That wire could be either one-enforced or two-enforced. Furthermore, if the wire is one-enforced, the cycle follows a particular parity of zigzag that is known in advance. When attaching a clause gadget to the horizontal variable wire, the vertical clause wire is aligned so that if the variable wire is one-enforced, the two wires match up as shown in Figure~\ref{fig:variable_clause_mismatch_1}.

\begin{figure}[!htbp]
    \centering
    \includegraphics[scale=0.5]{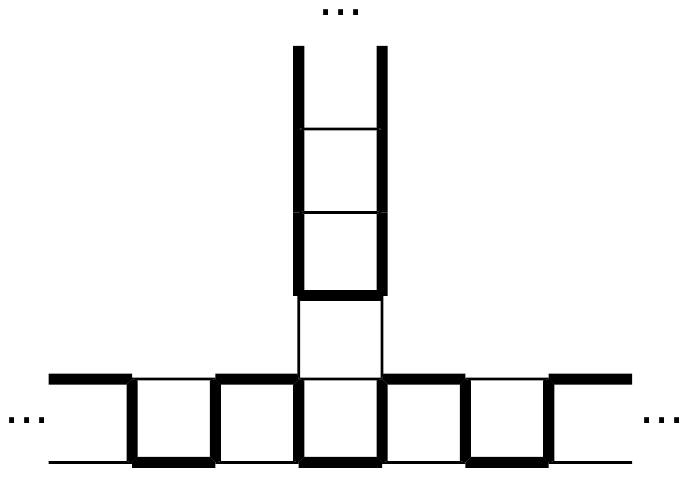}
    \includegraphics[scale=0.5]{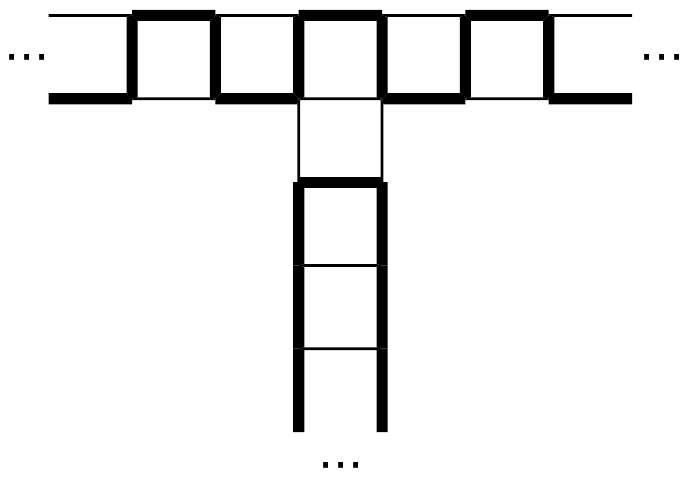}
    \caption{The junction between a variable gadget's horizontal wire and one of the three vertical wires from a clause gadgets; the case shown is where the variable gadget's wire is one-enforced; in this case, the cycle does not ``match up''}
    \label{fig:variable_clause_mismatch_1}
\end{figure}

We prove that if the variable wire is one-enforced, the situation from Figure~\ref{fig:variable_clause_mismatch_1} is exactly what will occur. Without loss of generality, suppose the horizontal wire is the top wire of the variable node. The cycle in the one-enforced variable wire propagates from the left through this wire. And the two arcs of the cycle in the two-enforced vertical clause wire propagate down. At this point, all we know is what is shown in Figure~\ref{fig:variable_clause_mismatch_2a}. Consider the circled node in that figure. It must contribute the next two edges as shown in Figure~\ref{fig:variable_clause_mismatch_2b}. Then the next circled node contributes the next edge as shown in Figure~\ref{fig:variable_clause_mismatch_2c}. Finally, the circled vertex in that figure contributes two edges to the cycle in the horizontal variable wire as shown in Figure~\ref{fig:variable_clause_mismatch_2d}. Thus the horizontal variable wire continues to be one-enforced despite the fact that the vertical wire is attached. As desired, the cycle must pass through this clause/variable junction as shown in Figure~\ref{fig:variable_clause_mismatch_1}.

\begin{figure}[!htbp]
    \centering
    \begin{subfigure}[b]{1.6in}
    \centering
    \includegraphics[scale=0.5]{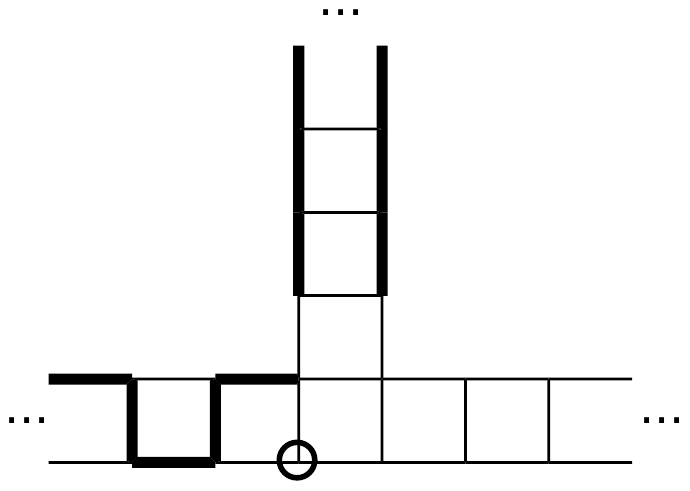} 
    \caption{}
    \label{fig:variable_clause_mismatch_2a}
    \end{subfigure}
    \begin{subfigure}[b]{1.6in}
    \centering
    \includegraphics[scale=0.5]{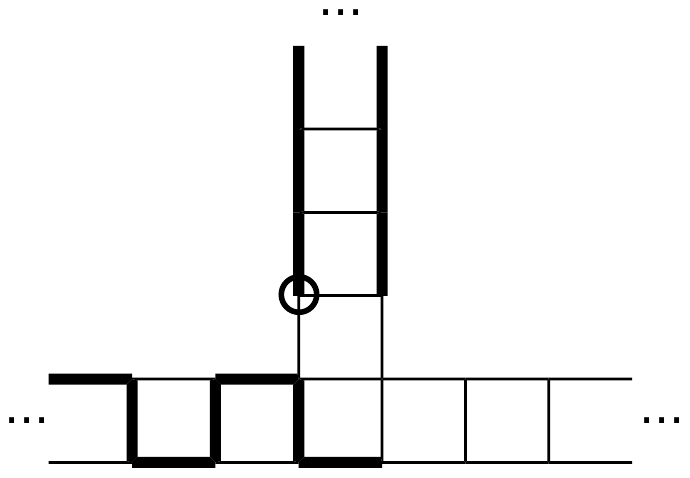} 
    \caption{}
    \label{fig:variable_clause_mismatch_2b}
    \end{subfigure}
    \begin{subfigure}[b]{1.6in}
    \centering
    \includegraphics[scale=0.5]{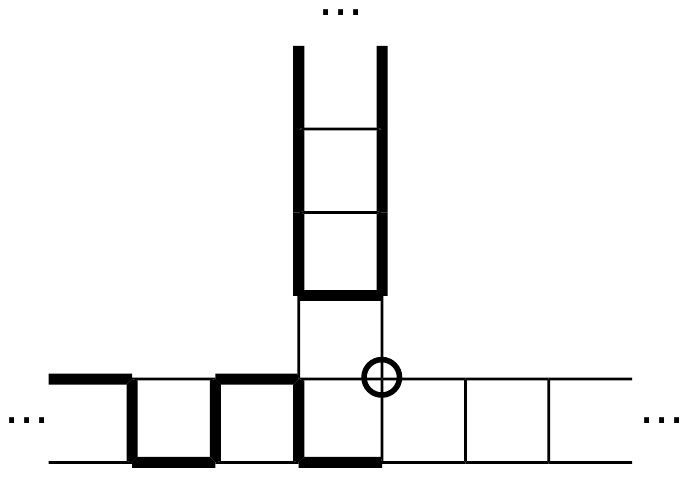} 
    \caption{}
    \label{fig:variable_clause_mismatch_2c}
    \end{subfigure}
    \begin{subfigure}[b]{1.6in}
    \centering
    \includegraphics[scale=0.5]{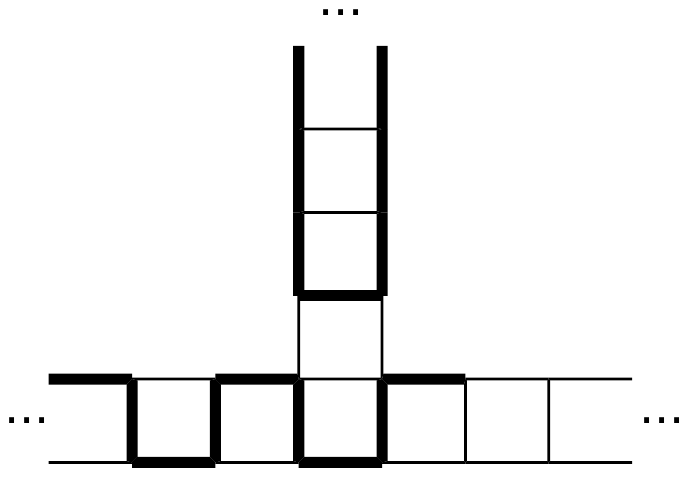} 
    \caption{}
    \label{fig:variable_clause_mismatch_2d}
    \end{subfigure}
    \caption{The parts of the cycle in the variable and clause wires do not connect if the variable wire is one-enforced}
    \label{fig:variable_clause_mismatch_2}
\end{figure}

On the other hand, when the horizontal variable wire is two-enforced, the parts of the cycle inside the variable and clause gadgets are free to connect as shown in Figure~\ref{fig:variable_clause_match}. Note however, that the parts of the cycle in the two gadgets are also free to not connect, which is useful when multiple variables satisfy a clause. 

\begin{figure}[!htbp]
    \centering
    \includegraphics[scale=0.5]{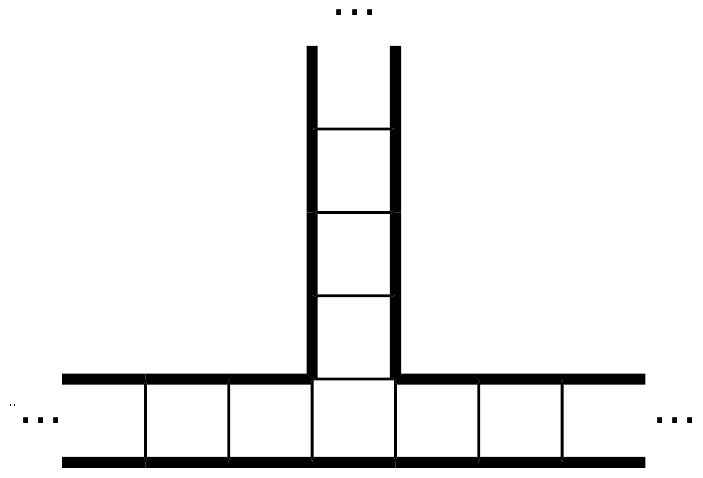}
    \caption{The parts of the cycle in the variable and clause wires can connect if the variable wire is two-enforced}
    \label{fig:variable_clause_match}
\end{figure}

Thus we can conclude the following two facts about our clause gadget. 

If a Hamiltonian cycle passes through each vertex gadget such that some clause gadget attaches to some three one-enforced wires, then the part of the cycle in the clause gadget has no way to connect to the part in the variable gadgets and so we don't actually have a Hamiltonian cycle (rather we have a contradiction). In other words, a Hamiltonian cycle passing through each vertex gadget must attach each clause gadget to at least one two-enforced horizontal variable wire.

If a grid graph has a Hamiltonian cycle and we add a clause gadget, attaching it to three horizontal vertex wires of which at least one is two-enforced, then the new graph is also Hamiltonian. This is because we can simply join the cycle that forms the boundary of the clause gadget to the pre-existing Hamiltonian cycle at one of the two-enforced horizontal variable wires (as in Figure~\ref{fig:variable_clause_match}), merging them into a Hamiltonian cycle for the entire graph.

\FloatBarrier
\subsection{Overall Reduction}
\FloatBarrier


Suppose we are given an instance of \prob{Planar Monotone Rectilinear 3SAT}, or in other words a monotone rectilinear embedding of a 3-CNF formula in the plane. We convert this instance into a grid graph as follows. Each variable segment gets replaced by a variable gadget, which are connected in order with wires and one-enforcers. Then a long wire (possibly with extra one-enforcers for parity) connects the two variable gadgets at the ends, forming a big loop of variable gadgets (and one-enforcers). After that, we add clause gadgets for each clause. The positive clauses get a clause gadget above the variable gadgets while the negative clauses get a clause gadget inside the loop. An example schematic of what the final result might be is shown in Figure~\ref{fig:polygonal_grid_graph}, corresponding to the \prob{Planar Monotone Rectilinear 3SAT} instance shown in Figure~\ref{fig:square_polygonal_example}.

\begin{figure}[!htbp]
    \centering
    \includegraphics[width=\textwidth]{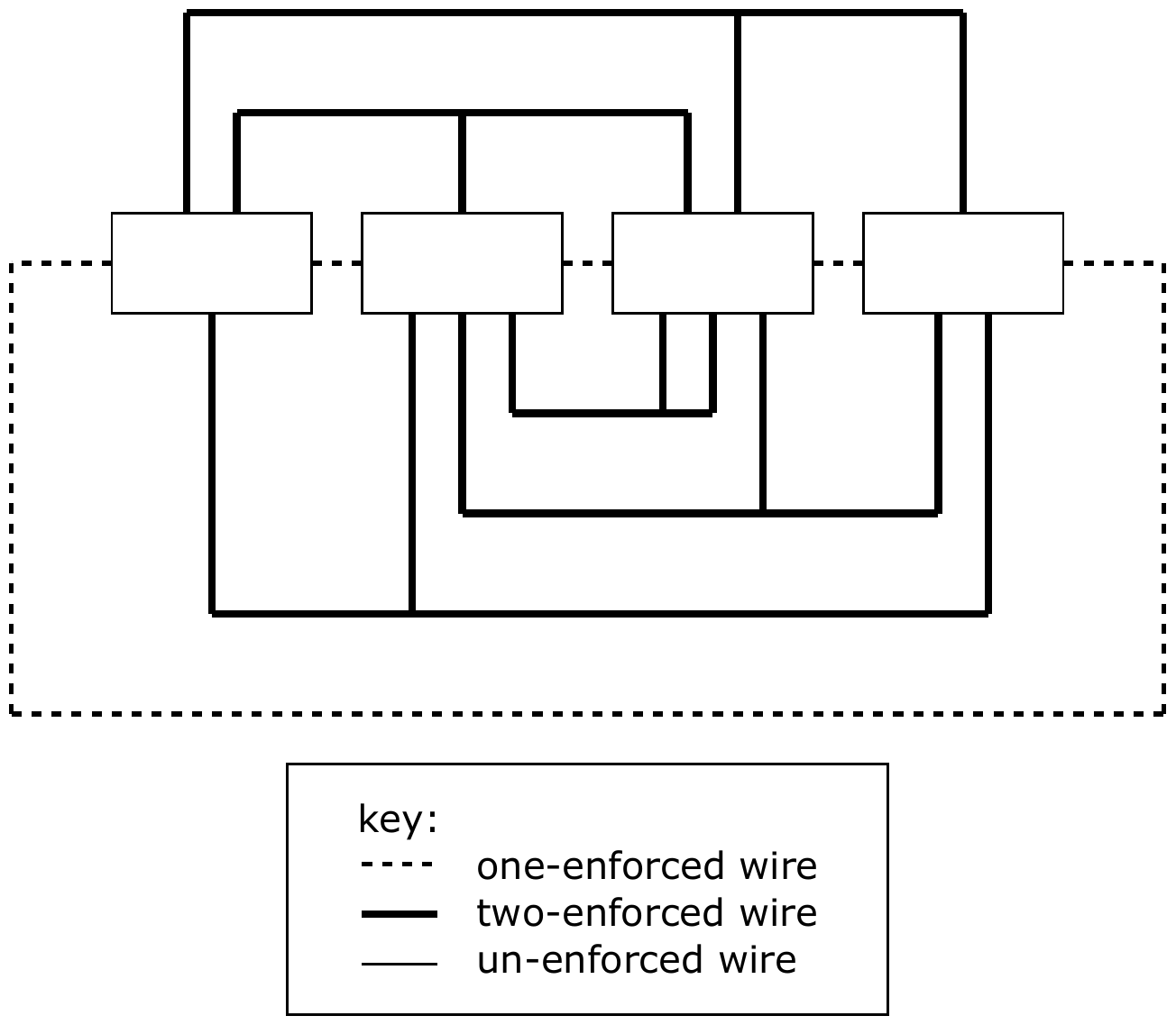}
    \caption{An example schematic for the resulting grid graph that could be produced from the instance of \prob{Planar Monotone Rectilinear 3SAT} in Figure~\ref{fig:square_polygonal_example}}
    \label{fig:polygonal_grid_graph}
\end{figure}

As you can see, the resulting instance of \prob{Hamiltonicity of Square Polygonal Grid Graphs} (the resulting grid graph) is very similar in shape to the original embedding. This similarity in shape between the embedding inputted into the reduction and the resulting instance of \prob{Hamiltonicity of Square Polygonal Grid Graphs} can be used as an argument that this is a polynomial time reduction. It is very easy to construct the \prob{Hamiltonicity of Square Polygonal Grid Graphs} instance from the embedding, and furthermore, the resulting graph is barely larger than the embedding. Due to the simplicity of the reduction, it is possible to construct the grid graph in time proportional to the size of the grid graph, but that size is itself linear in the size of the embedding. Thus we see that the reduction given above runs in polynomial time.

The remaining goal is to show that the reduction is answer preserving.

Consider first the case that the grid graph is Hamiltonian. Then we construct a variable assignment by assigning the value true to a variable if and only if that variable's variable gadget has the top wire two-enforced. Consider any [positive/negative] clause [$(x_i \lor x_j \lor x_k)$/$(\bar{x_i} \lor \bar{x_j} \lor \bar{x_k})$] under this assignment. The corresponding gadget for a [positive/negative] clause is [outside/inside] the large loop in the grid graph; thus the clause gadget is attached to the [top/bottom] wires of the variable gadgets for variables $x_i$, $x_j$, and $x_k$. By one of our derived properties for clause gadgets, we know that each clause is attached to at least one two-enforced wire; thus the [top/bottom] wire for the $x_i$, $x_j$, or $x_k$ variable gadget must be two-enforced. Equivalently, either $x_i$, $x_j$, or $x_k$ must be [true/false]. Thus, at least one variable in the clause is [true/false], and so at least one literal in the clause is true. In other words, the entire clause must be true. Then we see that we have found a satisfying assignment.

Next consider the case that a satisfying assignment for the formula exists. Consider the grid graph with all the clause gadgets removed (just a loop of variable gadgets and one-enforcers). Certainly a Hamiltonian cycle exists in this graph; in fact, many do: the cycle has two choices of behavior at each variable gadget. Construct the particular Hamiltonian cycle in which a variable gadget's top wire is two-enforced if and only if the variable is assigned a value of true. We will add the clause gadgets back into the graph one at a time. By one of our derived properties for clause gadgets, we know that if the clause being added attaches to at least one two-enforced wire then we can extend the Hamiltonian cycle to the new graph. Consider a clause gadget for a [positive/negative] clause [$(x_i \lor x_j \lor x_k)$/$(\bar{x_i} \lor \bar{x_j} \lor \bar{x_k})$]. The gadget is [outside/inside] the large loop in the grid graph and is therefore attached to the [top/bottom] wires of the variable gadgets for variables $x_i$, $x_j$, and $x_k$. But since the clause is satisfied, either $x_i$, $x_j$, or $x_k$ must be [true/false] and so the [top/bottom] wire for that variable gadget must be two-enforced. Thus the clause attaches to at least one two-enforced wire. Thus, as we add the clause gadgets back into the graph one at a time, the graph remains Hamiltonian. Clearly then the full graph is Hamiltonian.

As desired, we see that the grid graph in question is Hamiltonian if and only if the 3-CNF formula is satisfiable and that therefore we have described a polynomial-time answer-preserving reduction.

\FloatBarrier
\section{Conclusion and Further Work}
\FloatBarrier
\label{sec:conclusion}

We showed that \prob{Hamiltonicity of Thin Polygonal Grid Graphs} is solvable in polynomial time for every shape of grid graph.  In addition, we showed that \prob{Hamiltonicity of Square Polygonal Grid Graphs} and \prob{Hamiltonicity of Hexagonal Thin Grid Graphs} are both NP-complete. Having determined this, we have proved the hardness of two of the problems left open by Arkin et al.\ \cite{Arkin}, as well as determining the complexity of three problems not addressed in that paper.

This leaves only one of Arkin et al.'s open problems, namely, the complexity of \prob{Hamiltonicity of Hexagonal Solid Grid Graphs}.
Arkin et al.\ conjectured that this problem can be solved in polynomial time, based on the idea that the polynomial-time algorithm for \prob{Hamiltonicity of Square Solid Grid Graphs} could be adapted to also solve \prob{Hamiltonicity of Hexagonal Solid Grid Graphs}. 
The algorithm is a cycle-merging algorithm which starts with a 2-factor for the grid graph and progressively merges cycles until the 2-factor has one component (a Hamiltonian cycle) or until further merging is impossible. Many of the ideas appear relevant to hexagonal grid graphs as well, but a direct translation of the algorithm from solid square grid graphs and the relevant correctness proofs to hexagonal grid graphs fails. In order to make this approach work, some new insight seems necessary. Nevertheless, having struggled with this problem for some time, we believe the conjecture to be correct: \prob{Hamiltonicity of Hexagonal Solid Grid Graphs} can probably be solved in polynomial time.

\section*{Acknowledgments}
This research was initiated during the open problem sessions of MIT's graduate class 6.890: Algorithmic Lower Bounds, Fall 2014. We thank the other participants of those sessions---in particular, Quanquan Liu and Yun William Yu---for helpful discussions and for providing a stimulating research environment.

\bibliography{grid_graphs}

\begin{thebibliography}{10}

\bibitem{Andersson-2009}
Daniel Andersson.
\newblock Hashiwokakero is {NP}-complete.
\newblock {\em Information Processing Letters}, 109(19):1145--1146, 2009.

\bibitem{Arkin-Bender-Demaine-Fekete-Mitchell-Sethia-2005}
Esther~M. Arkin, Michael~A. Bender, Erik~D. Demaine, S\'andor~P. Fekete, Joseph
  S.~B. Mitchell, and Saurabh Sethia.
\newblock Optimal covering tours with turn costs.
\newblock {\em SIAM Journal on Computing}, 35(3):531--566, 2005.

\bibitem{Arkin}
Esther~M. Arkin, S{\'a}ndor~P. Fekete, Kamrul Islam, Henk Meijer, Joseph S.~B.
  Mitchell, Yurai N{\'u}{\~n}ez-Rodr{\'\i}guez, Valentin Polishchuk, David
  Rappaport, and Henry Xiao.
\newblock Not being (super)thin or solid is hard: A study of grid
  hamiltonicity.
\newblock {\em Computational Geometry: Theory and Applications},
  42(6--7):582--605, 2009.
\newblock Originally published at EuroComb'07.

\bibitem{Arkin-Fekete-Mitchell-2000}
Esther~M. Arkin, S{\'a}ndor~P. Fekete, and Joseph~S.B. Mitchell.
\newblock Approximation algorithms for lawn mowing and milling.
\newblock {\em Computational Geometry: Theory and Applications},
  17(1--2):25--50, 2000.

\bibitem{deBerg}
Mark Berg and Amirali Khosravi.
\newblock Optimal binary space partitions in the plane.
\newblock In {\em Proceedings of the 16th Annual International Conference on
  Computing and Combinatorics}, volume 6196 of {\em Lecture Notes in Computer
  Science}, pages 216--225, Nha Trang, Vietnam, July 2010.

\bibitem{6.890}
Erik~D. Demaine.
\newblock Lecture 8: Hamiltonicity.
\newblock In {\em MIT 6.890: Algorithmic Lower Bounds}. 2014.
\newblock \url{http://courses.csail.mit.edu/6.890/fall14/lectures/L08.html}.

\bibitem{trvb}
Erik~D. Demaine and Mikhail Rudoy.
\newblock Tree-residue vertex-breaking: a new tool for proving hardness.
\newblock arXiv:1706.07900, June 2017.
\newblock \url{http://arxiv.org/abs/1706.07900}.

\bibitem{Forisek-2010}
Michal Fori{\v{s}}ek.
\newblock Computational complexity of two-dimensional platform games.
\newblock In {\em Proceedings of the 5th International Conference on Fun with
  Algorithms}, pages 214--227, Ischia, Italy, June 2010.

\bibitem{Itai-Papadimitriou-Luiz-1982}
Alon Itai, Christos~H. Papadimitriou, and Jayme~Luiz Szwarcfiter.
\newblock Hamilton paths in grid graphs.
\newblock {\em SIAM Journal on Computing}, 11(4):676--686, November 1982.

\bibitem{Karp-1972}
Richard~M. Karp.
\newblock Reducibility among combinatorial problems.
\newblock In {\em Proceedings of a Symposium on the Complexity of Computer
  Computations}, pages 85--103, Yorktown Heights, New York, March 1972.

\bibitem{Papadimitriou-Vazirani-1984}
Christos~H. Papadimitriou and Umesh~V. Vazirani.
\newblock On two geometric problems related to the travelling salesman problem.
\newblock {\em Journal of Algorithms}, 5(2):231--246, 1984.

\bibitem{Umans-Lenhart-1997}
Christopher Umans and William Lenhart.
\newblock Hamiltonian cycles in solid grid graphs.
\newblock In {\em Proceedings of the 38th Annual IEEE Conference on Foundations
  of Computer Science}, pages 496--505, Miami, Florida, October 1997.

\bibitem{Yato-2000}
Takayuki Yato.
\newblock On the {NP}-completeness of the {S}lither {L}ink puzzle.
\newblock {\em IPSJ SiG Notes}, AL-74:25--32, 2000.

\end{thebibliography}
\bibliographystyle{plain}

\end{document}